\documentclass[12pt]{amsart}
\usepackage{amssymb}
\usepackage{amsfonts}
\usepackage{amscd}
\usepackage[all]{xy}
\usepackage{graphicx}
\usepackage[hidelinks]{hyperref}
\usepackage{centernot}
\usepackage{mathrsfs}
\usepackage{tikz}
\usetikzlibrary{arrows.meta, decorations.markings, calc}

\usepackage{enumerate}
\newcommand{\Conf}{\mathrm{Conf}}
\newcommand{\Aut}{\mathrm{Aut}}
\newcommand{\Diff}{\mathrm{Diff}}
\newcommand{\Homeo}{\mathrm{Homeo}}
\newcommand{\dom}{\mathrm{dom}}
\newcommand{\ran}{\mathrm{ran}}
\newcommand{\Stab}{\mathrm{Stab}}

\newtheorem{theorem}{Theorem}[section]
\newtheorem{lem}[theorem]{Lemma}

\newtheorem{prop}[theorem]{Proposition}
\newtheorem{ass}[theorem]{Assertion}

\theoremstyle{definition}
\newtheorem{definition}[theorem]{Definition}
\newtheorem{example}[theorem]{Example}

\newtheorem{def-prop}[theorem]{Definition-Proposition}

\theoremstyle{remark}
\newtheorem{remark}[theorem]{Remark}

\theoremstyle{plain}

\numberwithin{equation}{subsection}

\def\boxit#1#2{\setbox1=\hbox{\kern#1{#2}\kern#1}%
	\dimen1=\ht1 \advance\dimen1 by #1 \dimen2=\dp1
	\advance\dimen2 by #1
	\setbox1=\hbox{\vrule height\dimen1 depth\dimen2\box1\vrule}%
	\setbox1=\vbox{\hrule\box1\hrule}%
	\advance\dimen1 by .4pt \ht1=\dimen1 \advance\dimen2 by
	.4pt \dp1=\dimen2 \box1\relax}

\def\cM{{\mathcal M}}
\def\cN{{\mathcal N}}

\begin{document}

	\title{On the conformal group of a globally hyperbolic spacetime} 	

	\author{Ali Bleybel}

	\address{Faculty of Sciences (I), Lebanese University, Beirut, Lebanon}

	\begin{abstract}
		We study causal and conformal automorphism groups of globally hyperbolic spacetimes using an
		order-theoretic back-and-forth method on dense countable subsets. In two dimensions we show that
		any connected, globally hyperbolic spacetime with non-compact Cauchy surfaces that is directed is
		causally isomorphic to the Minkowski plane $\mathbb{M}^2$. Consequently, we obtain a partial
		classification of the causal and conformal automorphism groups of two-dimensional globally
		hyperbolic spacetimes, including the cases with compact Cauchy surfaces and non-directed causal
		order. The directed non-compact case is handled by refining the dense back-and-forth construction
		with the two intrinsic null orders, which record the two spacelike sides forgotten by bare causal
		incomparability. On the physics side, the resulting symmetry descriptions can
		be read as a factorized-versus-matched action of large reparametrization groups on null-type
		completion boundaries, illustrated by moving mirrors, conformal interfaces, and FLRW toy models.
	\end{abstract}

\maketitle

\section{Introduction}
\label{intro}
In General Relativity, conformal mappings are important because they preserve the causal structure (up to time orientation) and light-like geodesics up to parametrization.

Zeeman, \cite{Z}, showed that the group of causal automorphisms of Minkowski spacetime (of dimension $>2$) coincides with the group $G$ generated by orthochronous Lorentz transformations, space-time translations, and dilatations. Later, in \cite{L}, Levichev generalized these results to smooth time-orientable Lorentzian manifolds $M$ with dimensions $>2$. More precisely, he showed that any causal automorphism of $M$ is automatically conformal, provided that it is past and future distinguishing.

The causal and conformal automorphism groups of two-dimensional globally
hyperbolic spacetimes have been studied systematically by D.-H.~Kim
\cite{K,K2}. In particular, his results provide the foundational embedding
and subgroup framework for such groups, both in the non-compact Cauchy case
and in the compact case via the universal cover. The present article is
complementary to that framework. Its main purpose is to isolate an intrinsic
order-theoretic reconstruction mechanism in the directed non-compact case,
based on the two oriented null orders, and to use this mechanism to organize
the resulting automorphism and conformal symmetry statements in a form adapted
to dense-set reconstruction and boundary-data descriptions.

The paper addresses this problem using the back-and-forth method (reminiscent of model theory).
We obtain a partial result, namely for two-dimensional globally hyperbolic spacetimes. In particular, we isolate an order-theoretic condition satisfied by spacetimes causally isomorphic to two-dimensional Minkowski spacetime. These results might be of interest especially in the light of the celestial holography program \cite{CelestialCFT,StromingerReview}.

The point of the proof is not merely to rederive the smooth two-dimensional
classification. Rather, the argument isolates the order-theoretic data needed
for reconstruction: the two oriented null orders and their finite rectangles.
This formulation avoids a premature reliance on smooth null coordinates and
makes clear which parts of the proof are causal-order/topological in nature.
Although the present paper remains in the category of smooth globally
hyperbolic spacetimes, such a formulation is useful in view of recent synthetic
and low-regularity Lorentzian frameworks, such as Lorentzian length spaces,
where one seeks to formulate causality and reconstruction statements in terms
of causal and time-separation data rather than smooth charts
\cite{KunzingerSaemann}. We do not pursue such extensions here.

These questions are conceptually related to the asymptotic symmetry analysis of asymptotically flat spacetimes, where the Bondi--Metzner--Sachs (BMS) group acts at null infinity \cite{BondiMetzner,Sachs,StromingerReview}. Recent work in celestial holography proposes an equivalence between gravitational S-matrix elements and correlators in a conformal field theory living on this boundary \cite{CelestialCFT}. While our results are strictly $1+1$-dimensional and do not address higher-dimensional BMS geometry directly, they provide a simple order-theoretic toy model in which large reparametrization actions on null-type boundaries can be analyzed rigorously.

\paragraph{Contributions.}
\begin{enumerate}[1.]
	\item We formulate a back-and-forth criterion for producing causal isomorphisms from dense
	countable null-free subsets of globally hyperbolic spacetimes, and we record a concrete extension
	argument from a dense causal isomorphism to a spacetime causal isomorphism.
	\item In $1+1$ dimensions we construct two intrinsic null orders from the four-region decomposition
	of the complement of a null cone, and use finite null rectangles for the one-point extension step.
	\item As an application we give explicit directed-case identifications and
	group bounds for $\Aut(\mathcal M)$ and $\Conf(\mathcal M)$ in the compact
	and non-directed cases, organized by the corresponding boundary or
	universal-cover data.
\end{enumerate}
The rest of this paper is organized as follows. Section~2 outlines the background material and
fixes notation. Section~3 develops the back-and-forth framework on dense null-free subsets and
explains how a dense causal isomorphism extends to a spacetime causal isomorphism in the globally
hyperbolic setting. Section~4 applies this framework in dimension $1+1$: the two intrinsic null
orders provide the admissible finite maps and the one-point extension step, yielding a causal
isomorphism to $\mathbb{M}^2$ in the directed non-compact case and leading to a partial
classification of $\Aut(\mathcal M)$ and $\Conf(\mathcal M)$ in the remaining cases. Section~5 discusses consequences for symmetry actions
on null boundaries in $1+1$ dimensions. 

A recurring theme in the $1+1$-dimensional applications is that the causal/conformal symmetry descriptions have a direct boundary-symmetry interpretation. In the directed regime one obtains factorized reparametrization actions on the two null directions, whereas in the non-directed regime distinguished completion-boundary data force a matched action governed by a monotone boundary relation. Section~5 develops this viewpoint through moving mirrors, conformal interfaces, and cosmological toy models. Section~\ref{Examples} gives explicit compact examples, namely the Einstein static cylinder and global de Sitter space, and Section~\ref{conc} contains the conclusion.

\section{Preliminaries}
We assume the reader is familiar with elementary first-order logic and model theory.

\subsection{Order theory}
Let $(X, \prec)$ be a partially ordered set (\textit{poset}), with the order relation denoted by $\prec$. Recall that $\prec$, a binary relation, is a subset of the Cartesian product $X \times X$. Being an order relation, $\prec$ is a reflexive, transitive, and antisymmetric relation on $X$. If $A$ is a subset of $X$, the {\it induced order} on $A$ is the binary relation $\prec|_A := (A \times A) \cap \prec$. $\prec|_A$ is an order relation on $A$.

Given (partially) ordered sets $X, Y$, a partial isomorphism $X \dashrightarrow Y$ is an order isomorphism $A \to B$ with $A \subset X$, $B \subset Y$ equipped with the induced orders from $X$ and $Y$, respectively.

We say that $X$ is upward directed if, for all $x, y \in X$, there exists some $u \in X$ such that $x \prec u \,\wedge\, y \prec u$. Similarly, $X$ is downward directed if, for all $x, y \in X$, there exists some $v \in X$ such that $v \prec x \,\wedge\, v \prec y$.
A poset $(X, \prec)$ is said directed if it is both upward and downward directed.

The set $\{ x \in X \mid a \prec x \prec b\}$ will be denoted by $[a,b]$ and is called the Alexandrov interval with endpoints $a, b$.
For a spacetime $\cM$, an interval $[a,b]$ is also called a (closed) causal diamond.

By a \textit{pattern}, we mean a finite subset of $X$, equipped with a strict order relation.

\subsection{The setting}
We will be using the following signature
$$ L := \{ \twoheadrightarrow, \ll\}, $$
where $\twoheadrightarrow, \ll$ are two binary relation symbols.

\subsubsection{Theory}\label{axioms}
We consider the theory $T$, which consists of all statements that are logical consequences of the following axioms:
\begin{enumerate}
	\item $ \forall x \, ( x \centernot \ll x )$
	\item $ \forall x \forall y \forall z \; (( x \ll y \; \wedge \; y \ll z) \rightarrow x \ll z)$
	\item $ \forall x \, (x \twoheadrightarrow x) $
	\item $ \forall x \forall y \, ((x \twoheadrightarrow y \, \wedge \, y \twoheadrightarrow x) \rightarrow x=y) $
	\item $ \forall x \forall y \, \neg (x \ll y \, \wedge \, x \twoheadrightarrow y).  $
\end{enumerate}
And for each natural number $n > 0$, the following axiom:
{}\hspace{0.1cm}(6)$_n$
\begin{equation*}
	\forall x_1,\dots,x_n\left[\Big(\bigwedge_{i=1}^{n-1}(x_i\twoheadrightarrow x_{i+1})\wedge x_1\twoheadrightarrow x_n \Big)
	\longrightarrow \bigwedge_{1\leq j<k\leq n} (x_j\twoheadrightarrow x_k)\right].
\end{equation*}
The axiom schema $(6)_n$ expresses a known, standard fact about Lorentzian spacetimes: if consecutive points along a chain are null-related and the endpoints are also null-related, then all intermediate pairs must be null-related as well (null relatedness ``propagates'' along such chains).

\subsubsection{Types and quantifier-free types}
An $n$-\textit{type} $t = t(\bar{x})$ in $T$ is a set of formulas $\phi(\bar{x})$ that is realized in some model of $T$; i.e., there exists a model $M$ of $T$ and a tuple $\bar{a} \in M^n$ such that $M \models \phi(\bar{a})$ for every $\phi(\bar{x}) \in t$.
The case where $n=1$ corresponds to $1$-types.

We say that $t(\bar{x})$ is a quantifier-free type if all formulas in $t$ are quantifier-free.

For a tuple $\bar{a} \in M^n$, $n >0$, we denote by tp$(\bar{a})$ the set of all formulas $\phi(\bar{x})$ such that $\phi(\bar{a})$ holds in $M$. Let $X$ be a set of elements of $M$, and let $\bar{b}$ be a listing of the elements of $X$ (so $\bar{b}$ can be of infinite length). The \textit{type of $\bar{a}$ over $X$}, denoted by tp$(\bar{a}/\bar{b})$ (or, equivalently, tp$(\bar{a}/X)$), is the set of all formulas $\phi(\bar{x}, \bar{y})$ (where $\bar{y}$ is a tuple of variables), such that $\phi(\bar{a}, \bar{b}')$ holds in $M$. Here $\bar{b}'$ is some tuple of elements of $X$ of the same length as $\bar{y}$.

The \textit{quantifier-free type of $\bar{a}$ over $X$}, denoted by qftp$(\bar{a}/X)$, is the set of all quantifier-free formulas $\phi(\bar{x}, \bar{y})$ such that $\phi(\bar{a}, \bar{b}')$ holds in $M$.

\subsection{Corresponding notions in the context of spacetime physics}\label{spacetime-physics}
The causality relation in the context of spacetime physics is denoted by $\prec$: $ x \prec y$ iff $x \in J^-(y)$ where $J^-(y)$ is the past causal cone of $y$; note that $y \in J^-(y)$.
We will also make use of the relation $\ll$: $x \ll y$ (Chronology) iff $x \in I^-(y)$ ($x$ belongs to the \textit{open} past lightcone of $y$). For a globally hyperbolic spacetime, the relation $\ll$ is transitive and irreflexive; it follows that $\ll$ is antisymmetric. The relation $\ll$ is then a strict order.

Let $X \subset \cM$ be a finite set. Then $X$ is a \textit{strict pattern} if for all $x, y \in X$, $x \neq y$, $x \ll  y$ or $x \nprec y$. The point of this definition is to exclude non-open (null) relations; equivalently, a subset is a strict pattern precisely when it is null-free in the sense of Lemma~\ref{null-free}.
The statements $x \ll y$ and $I^-(y) \cap I^+(x) \neq \varnothing$ are equivalent. In particular, $x\ll y$ implies $I^+(y)\subsetneq I^+(x)$ and $I^-(x)\subsetneq I^-(y)$.

The expression $x \twoheadrightarrow y$ (used earlier in the context of the theory $T$) is defined here as $x \prec y \; \wedge \; x \centernot \ll y$.

We denote by $x^\perp$ the set of elements causally unrelated to $x$, i.e. $\{ y \in \cM \mid  x \centernot \prec y \,\wedge\, y \centernot \prec x\}$. The relation $ y \in x^\perp$ is also denoted by $ x \| y$.

\begin{ass}
Assume $a \twoheadrightarrow b$. Then, for any $c$, we have
\begin{eqnarray*}
	a \prec c \, \wedge \,  c \twoheadrightarrow b & \rightarrow &  a \twoheadrightarrow c   \\
	b \ll c  & \rightarrow & a \ll c  \\
	b || c & \rightarrow & c \nprec a
\end{eqnarray*}
\end{ass}

\begin{proof}[Proof sketch]
For the first claim: if $a \prec c$ and $c \twoheadrightarrow b$, then $a \prec c \prec b$. If $a \ll c$ held, then since $c \twoheadrightarrow b$ means $c \prec b$ and $c \centernot \ll b$, we would get $a \ll c \prec b$, hence $a \ll b$; but $a \twoheadrightarrow b$ gives $a \centernot \ll b$, so $a \centernot \ll c$, hence $a \twoheadrightarrow c$.
For the second: $a \twoheadrightarrow b$ gives $a \prec b$, and $b \ll c$; hence $a \prec b \ll c$ implies $a \ll c$ (since the composition of a causal relation followed by a chronological one is chronological).
For the third: if $c \prec a$, then $c \prec a \prec b$ (using $a \twoheadrightarrow b$), contradicting $b \| c$.
\end{proof}

\begin{remark}
The relation $x \prec y$ can be expressed as (modulo $T$):
$$ x \prec y  \equiv x \ll y \; \vee \; x \twoheadrightarrow y $$
where we used the following axiom (5):
$$ \neg (x \ll y \, \wedge \, x \twoheadrightarrow y).  $$
To see this, observe that $x \twoheadrightarrow y \leftrightarrow (x \prec y \, \wedge \, x \centernot \ll y)$, and the required equivalence follows.
\end{remark}

\begin{theorem}
Let $\mathcal{M}$ be a globally hyperbolic spacetime. Then $\mathcal{M}$ is a model of the theory $T$ in the language $L$.
\end{theorem}

\begin{proof}
It is a standard fact that the relations $\twoheadrightarrow$ and $\ll$ defined as in subsection \ref{spacetime-physics}, satisfy the axioms (1) through (5), and (6)$_n$ of \ref{axioms} for all $n>0$. For details, please see \cite{KP}.
\end{proof}

\subsection{Conformal transformations}
A conformal diffeomorphism $\phi: \mathcal{M} \to \mathcal{N}$ where $(\mathcal{M}, g_\mathcal{M})$ and $(\mathcal{N}, g_\mathcal{N})$ are two $n$-dimensional pseudo-Riemannian manifolds of signature $(p,q)$, $0 \leq p \leq q$, $p+q = n$, is a $C^\infty$-diffeomorphism
$\varphi: \mathcal{M} \to \mathcal{N}$ such that
$$\varphi^*(g_\mathcal{N}) = e^\psi g_\mathcal{M}, $$
for some arbitrary $C^\infty$-function $\psi: \mathcal{M} \to \mathbb{R}$.

\section{Causal isomorphisms of globally hyperbolic spacetimes}

\subsection{General considerations}
Let $\mathcal{C}, \mathcal{D}$ be models of $T$.
A causal isomorphism $ f: \mathcal{C} \to \mathcal{D}$ is
a bijective map satisfying
$$ \forall x, y \in \mathcal{C}, x \prec y   \leftrightarrow f(x) \prec f(y).$$
The above definition applies to globally hyperbolic spacetimes $\mathcal{M}, \mathcal{N}$.

\begin{lem}\label{null-free}
Let $\mathcal{M}$ be a globally hyperbolic spacetime. Then there exists a countable dense subset $C \subset \mathcal{M}$ such that for all distinct $p, q \in C$,
\[  \neg(p \twoheadrightarrow q) \wedge \neg( q \twoheadrightarrow p). \]
Equivalently, there are no null-related pairs in $C$.

In particular, for any $a\in \mathcal{M}$ there exists a countable dense null-free $C$ with $a\in C$.
\end{lem}

\begin{proof}
First let \((U_n)_{n\in\mathbb N}\) be a fixed countable basis of neighbourhoods of
\(\mathcal M\). We construct a sequence \((c_n)_{n\in\mathbb N}\) inductively. Suppose
\(c_0,\dots,c_{n-1}\) have been chosen. Put
\[
F_n=\bigcup_{m<n}\bigl(\partial I^+(c_m)\cup\partial I^-(c_m)\bigr)
      \cup\{c_0,\dots,c_{n-1}\}.
\]
Each set \(\partial I^\pm(c_m)\) is closed with empty interior: in a globally hyperbolic
spacetime it is achronal, and an achronal subset cannot contain a nonempty open set.
Since \(F_n\) is a finite union of such sets and of finitely many points, \(F_n\) has
empty interior. Hence \(U_n\setminus F_n\neq\varnothing\), and we may choose
\(c_n\in U_n\setminus F_n\).

Then \(C:=\{c_n:n\in\mathbb N\}\) meets every basis element, so it is dense. Moreover,
if \(m<n\), then \(c_n\notin\partial I^\pm(c_m)\); equivalently, \(c_m\) and \(c_n\)
are not null-related. Thus \(C\) is null-free.

Finally fix \(a\in\mathcal M\). Start the construction with \(c_0=a\), and enumerate
a countable basis as \((U_n)_{n\ge 1}\) after this initial choice. At the \(n\)-th
stage, \(n\ge1\), use the same avoidance set
\[
F_n=\bigcup_{m<n}\bigl(\partial I^+(c_m)\cup\partial I^-(c_m)\bigr)
      \cup\{c_0,\dots,c_{n-1}\}.
\]
The same argument gives \(c_n\in U_n\setminus F_n\). The resulting countable dense
null-free set contains \(a\) by construction.
\end{proof}

\begin{definition}[Back-and-forth system]\label{def:bf-system}
Let $C\subset\mathcal M$ and $D\subset\mathcal N$ be sets equipped with the induced relations
$\ll$ and $\perp$ (causal incomparability). A \emph{back-and-forth system} between $C$ and $D$ is a
nonempty family $\mathscr P$ of finite partial isomorphisms $p:A\to B$ (with $A\subset C$, $B\subset D$)
such that:
\begin{enumerate}[(i)]
	\item (\emph{Forth}) for every $p\in\mathscr P$ and every $c\in C$ there exists $q\in\mathscr P$
	extending $p$ with $c\in\dom(q)$;
	\item (\emph{Back}) for every $p\in\mathscr P$ and every $d\in D$ there exists $q\in\mathscr P$
	extending $p$ with $d\in\ran(q)$.
\end{enumerate}
\end{definition}

Let $a, b$ be points in  $\mathcal{M}, \mathcal{N}$ respectively.

\begin{theorem}[Back-and-forth criterion]\label{main-0}
Let $\mathcal M$ and $\mathcal N$ be globally hyperbolic spacetimes and fix $a\in\mathcal M$ and
$b\in\mathcal N$. Then there exists a causal isomorphism $\Phi:\mathcal M\to\mathcal N$ with
$\Phi(a)=b$ if and only if there exist countable dense null-free sets $C\subset\mathcal M$ and
$D\subset\mathcal N$ with $a\in C$, $b\in D$, and a back-and-forth system $\mathscr P$ between $C$
and $D$ (Definition~\ref{def:bf-system}) containing the map $\{a\}\to\{b\}$.
\end{theorem}

\begin{proof}

\noindent(\(\Rightarrow\))
Let \(\Phi:\mathcal M\to\mathcal N\) be a causal isomorphism with
\(\Phi(a)=b\). Choose a countable dense null-free set
\(C\subset\mathcal M\) with \(a\in C\) by Lemma~\ref{null-free}, and put
\[
D:=\Phi(C).
\]
Then \(D\) is countable and dense in \(\mathcal N\), contains \(b\), and is
null-free, since \(\Phi\) preserves the causal order. Let \(\mathscr P\) be the
family of all finite restrictions \(\Phi|_A\), where \(A\subset C\) is finite.
Then \(\mathscr P\) is a back-and-forth system: if \(p=\Phi|_A\) and \(c\in C\),
 then \(\Phi|_{A\cup \{c\}}\) gives the forth extension; if \(d\in D=\Phi(C)\),
write \(d=\Phi(c)\) and use the same extension for the back step. Moreover,
\(\mathscr P\) contains the map \(\{a\}\to\{b\}\).

\medskip
($\Leftarrow$) Assume we have countable dense null-free sets $C\subset\mathcal M$ and
$D\subset\mathcal N$ and a back-and-forth system $\mathscr P$ containing $\{a\}\to\{b\}$.
Fix enumerations
\[
C=\{c_0,c_1,c_2,\dots\},\qquad D=\{d_0,d_1,d_2,\dots\},
\]
with $c_0=a$ and $d_0=b$.

Using the back-and-forth property, we build an increasing chain
\[
p_0\subset p_1\subset p_2\subset\cdots
\]
with each $p_n\in\mathscr P$ such that for all $m$:
\[
c_m\in\dom(p_{2m})\quad\text{and}\quad d_m\in\ran(p_{2m+1}).
\]
Start with $p_0:\{a\}\to\{b\}$. Given $p_{2m}\in\mathscr P$, apply the \emph{back} clause to extend to
$p_{2m+1}\in\mathscr P$ with $d_m\in\ran(p_{2m+1})$. Given $p_{2m+1}\in\mathscr P$, apply the
\emph{forth} clause to extend to $p_{2m+2}\in\mathscr P$ with $c_{m+1}\in\dom(p_{2m+2})$.

Let \(p:=\bigcup_{n\ge 0}p_n\). Then \(p:C\to D\) is a well-defined bijection.
It preserves and reflects the chronological relation \(\ll\), because each
finite stage \(p_n\) does. Since \(C\) and \(D\) are null-free, the causal
relation on each of them is determined, for distinct points, by the chronological
relation; equality is preserved by injectivity. Hence \(p\) is a causal
isomorphism between \(C\) and \(D\).

Finally, since \(\mathcal M\) and \(\mathcal N\) are globally hyperbolic and
\(C,D\) are dense null-free subsets, the extension lemma proved below
(Lemma~\ref{lem:Phi}) extends \(p\) uniquely to a causal isomorphism
\[
\Phi:\mathcal M\longrightarrow\mathcal N
\]
with \(\Phi(a)=b\).
\end{proof}

Before proceeding, let us recall a standard result:
\begin{lem}
In a globally hyperbolic spacetime, the manifold topology coincides with the Alexandrov topology generated by sets of the form $I^{+}(p) \cap I^{-}(q)$.
\end{lem}

\begin{proof}
See Hawking et al.\ \cite{H}, particularly Section~3; for the related fact that the chronological order determines the manifold topology, see Malament \cite{M}.
\end{proof}

\begin{lem}\label{lem:Phi}
Let \(\cM,\cN\) be globally hyperbolic spacetimes, and let
\(C\subset\cM\), \(D\subset\cN\) be countable dense null-free subsets. Let
\(p:C\to D\) be a bijection such that, for all \(c_1,c_2\in C\),
\[
c_1\ll_{\cM} c_2
\quad\Longleftrightarrow\quad
p(c_1)\ll_{\cN}p(c_2).
\]
Then \(p\) extends uniquely to a causal isomorphism
\[
\Phi:\cM\longrightarrow \cN .
\]
\end{lem}

\begin{proof}
	Since \(C\) and \(D\) are null-free, the chronological and causal relations
	agree on distinct points of each set. Hence \(p\) is an isomorphism of the
    induced chronological orders \((C,\ll)\cong(D,\ll)\).
    	
    By the reconstruction theorem of Martin--Panangaden
    \cite[Theorem~7.2]{MP}, if \(C\) is a countable dense subset of a globally
    hyperbolic spacetime, equipped with the restricted chronological relation
    \(\ll=I^+\), then the spacetime is recovered, with its topology and causal
    order, as the space \(\max(I_C)\) of maximal elements of the ideal completion
    of the interval abstract basis
    \[
    \operatorname{int}(C)= \{(a,b)\in C^2 : a\ll b \}.
    \]
    The construction depends only on the induced chronological order on \(C\).
	
	The order isomorphism \(p:(C,\ll)\to(D,\ll)\) therefore induces an isomorphism
	of interval abstract bases
	\[
	\operatorname{int}(C)\cong \operatorname{int}(D),
	\]
	hence an isomorphism of their ideal completions
	\[
	I_C\cong I_D,
	\]
	and consequently a homeomorphism
	\[
	\Phi:\mathcal M\cong\max(I_C)\longrightarrow \max(I_D)\cong\mathcal N
	\]
	preserving the causal order. By construction, \(\Phi|_C=p\).
		
	Uniqueness follows from density: any causal isomorphism extending \(p\) agrees
	with \(\Phi\) on the dense set \(C\), and \(\mathcal N\) is Hausdorff.
\end{proof}

\section{Partial classification of \texorpdfstring{$\Aut(\mathcal{M})$}{Aut(M)} for a two-dimensional globally hyperbolic spacetime}
This section considers the case of a two-dimensional globally hyperbolic spacetime.

We have:
\begin{theorem}\label{main}
Let $\mathcal{M}$ be a connected two-dimensional globally hyperbolic spacetime having non-compact Cauchy surfaces.

Assume furthermore that $\mathcal{M}$ is directed.
Then a causal isomorphism exists $i: \mathcal{M} \to \mathbb{M}^2$.
Consequently,
$$\Aut(\mathcal{M}) \simeq \Aut(\mathbb{M}^2), $$
(via conjugation by $i$).
\end{theorem}

\subsection{Intrinsic null orders and the corrected back-and-forth setup}

Throughout this subsection, let $\mathcal M$ be a connected two-dimensional globally
hyperbolic spacetime with noncompact Cauchy surfaces, and assume that $\mathcal M$
is directed. Fix a time orientation and a smooth Cauchy temporal function
\[
\tau:\mathcal M\longrightarrow \mathbb R,
\qquad
\Sigma_s=\tau^{-1}(s),
\]
so that each $\Sigma_s$ is a smooth Cauchy line. We also fix an orientation of the
Cauchy lines, transported along a Bernal--S\'anchez splitting~\cite{BernalSanchez,BernalSanchez2}
$\mathcal M\simeq\mathbb R\times\Sigma_0$.

For $x\in\mathcal M$, write
\[
\mathcal N(x)=
\bigl(J^+(x)\setminus I^+(x)\bigr)
\cup
\bigl(J^-(x)\setminus I^-(x)\bigr)
\]
for the full null cone through $x$.

We shall use the two-dimensional null cone only in the following intrinsic form.
At each point there are two future null half-lines. The chosen orientation of the
Cauchy lines labels the corresponding null line fields by $+$ and $-$. Let
$\mathcal L_+$ and $\mathcal L_-$ denote the two families of maximal inextendible
null curves tangent to these line fields. Each leaf of either family meets every
Cauchy line exactly once, and hence each leaf is ordered by its intersection with
any fixed oriented Cauchy line.

\begin{lem}[Four-region decomposition]\label{lem:four-regions}
For every $x\in\mathcal M$, the complement
\[
\mathcal M\setminus \mathcal N(x)
\]
has exactly four open connected components:
\[
I^+(x),\qquad I^-(x),\qquad S_R(x),\qquad S_L(x).
\]
Here $S_R(x)$ and $S_L(x)$ are the two spacelike components, labelled by the
chosen orientation of the Cauchy lines.
\end{lem}

\begin{proof}
	Let $\tau:\mathcal M\to\mathbb R$ be the fixed Cauchy temporal function and
	write $\Sigma_s=\tau^{-1}(s)$. Using the chosen orientation of the Cauchy
	lines, identify the splitting with
	\[
	\mathcal M\simeq \mathbb R_s\times \mathbb R_r,
	\qquad
	\Sigma_s\simeq \{s\}\times\mathbb R.
	\]
	Let $\tau(x)=s_0$.
	
	For $s>s_0$, global hyperbolicity implies that
	\[
	J^+(x)\cap\Sigma_s
	\]
	is a compact interval in the Cauchy line $\Sigma_s$. Its two endpoints lie on
	the two future null generators from $x$. For $s<s_0$, the analogous statement
	holds for
	\[
	J^-(x)\cap\Sigma_s.
	\]
	At $s=s_0$, the null cone meets $\Sigma_{s_0}$ only at $x$.
	
	Thus the full null cone through $x$ is represented, in the Cauchy splitting, by
	two continuous graph curves
	\[
	r=\ell_x(s),\qquad r=r_x(s),
	\]
	with
	\[
	\ell_x(s)\leq r_x(s),
	\]
	and equality only at $s=s_0$. The continuity of $\ell_x$ and $r_x$ follows from
	the continuous dependence of the null generators, equivalently from the
	continuity of the endpoints of the compact intervals
	$J^\pm(x)\cap\Sigma_s$.
	
	Now define
	\[
	S_L(x)=\{(s,r):r<\ell_x(s)\},
	\qquad
	S_R(x)=\{(s,r):r>r_x(s)\}.
	\]
	These sets are open. Moreover,
	\[
	(s,r)\longmapsto (s,r-\ell_x(s))
	\]
	identifies $S_L(x)$ homeomorphically with
	\[
	\mathbb R\times(-\infty,0),
	\]
	and
	\[
	(s,r)\longmapsto (s,r-r_x(s))
	\]
	identifies $S_R(x)$ homeomorphically with
	\[
	\mathbb R\times(0,\infty).
	\]
	Hence $S_L(x)$ and $S_R(x)$ are connected.
	
	The remaining two components are the middle intervals:
	\[
	I^+(x)=\{(s,r):s>s_0,\ \ell_x(s)<r<r_x(s)\},
	\]
	and
	\[
	I^-(x)=\{(s,r):s<s_0,\ \ell_x(s)<r<r_x(s)\}.
	\]
	They are open and connected, since each is homeomorphic to
	\[
	(s_0,\infty)\times(0,1)
	\]
	or
	\[
	(-\infty,s_0)\times(0,1),
	\]
	respectively, using the affine parameter in the interval
	$(\ell_x(s),r_x(s))$.
	
	Therefore
	\[
	\mathcal M\setminus\mathcal N(x)
	=
	I^+(x)\sqcup I^-(x)\sqcup S_L(x)\sqcup S_R(x),
	\]
	and these four sets are precisely the four open connected components of the
	complement of the full null cone.
\end{proof}

\begin{lem}[Full null grid]\label{lem:full-null-grid}
	Assume that $\mathcal M$ is directed, in the sense that any two points have a
	common chronological future and a common chronological past. Then every 
	$+$-null leaf meets every $-$-null leaf in exactly one point. Consequently the
	intersection map
	\[
	\mathcal L_+\times\mathcal L_-\longrightarrow\mathcal M,
	\qquad
	(\lambda_+,\lambda_-)\longmapsto \lambda_+ \cap \lambda_-
	\]
	is bijective.
\end{lem}

\begin{proof}
	Let $\lambda_+$ be a $+$-null leaf and $\lambda_-$ a $-$-null leaf. Since
	each inextendible null leaf meets every Cauchy line exactly once, we may write
	them in the Cauchy splitting as continuous graphs
	\[
	r=\alpha(s),\qquad r=\beta(s),
	\]
	where $\alpha$ denotes the $+$-leaf and $\beta$ the $-$-leaf.
	
	First, they meet at most once. If two opposite null leaves met at two distinct
	points, the two null segments between the intersection points would form a null
	bigon. The interior of such a bigon contains timelike related points between
	the two vertices, contradicting the achronality of null boundary generators.
	
	It remains to prove existence of an intersection. Fix $s_0$ and set
	\[
	p=(s_0,\alpha(s_0)),\qquad q=(s_0,\beta(s_0)).
	\]
	
	Suppose first that
	\[
	\alpha(s_0)<\beta(s_0).
	\]
	By future-directedness, choose
	\[
	z\in I^+(p)\cap I^+(q),
	\qquad \tau(z)=s_1>s_0.
	\]
	On the Cauchy line $\Sigma_{s_1}$, the future of $p$ lies on or to the left of
	the $+$-null boundary through $p$, namely $\lambda_+$, while the future of
	$q$ lies on or to the right of the $-$-null boundary through $q$, namely
	$\lambda_-$. Since $z$ is a chronological future of both $p$ and $q$, it lies
	strictly between these two boundary points. Hence
	\[
	\beta(s_1)<\alpha(s_1).
	\]
	Thus the continuous function $\alpha-\beta$ changes sign between $s_0$ and
	$s_1$, and therefore $\lambda_+$ and $\lambda_-$ meet.
	
	Now suppose instead that
	\[
	\beta(s_0)<\alpha(s_0).
	\]
	By past-directedness, choose
	\[
	z\in I^-(p)\cap I^-(q),
	\qquad \tau(z)=s_{-1}<s_0.
	\]
	On $\Sigma_{s_{-1}}$, the past of $p$ lies on or to the right of the past
	segment of the $+$-null leaf $\lambda_+$, while the past of $q$ lies on or to
	the left of the past segment of the $-$-null leaf $\lambda_-$. Since $z$ is a
	chronological past of both $p$ and $q$, it lies strictly between these two
	boundary points. Hence
	\[
	\alpha(s_{-1})<\beta(s_{-1}).
	\]
	Thus $\alpha-\beta$ again changes sign, this time between $s_{-1}$ and $s_0$.
	So $\lambda_+$ and $\lambda_-$ meet.
	
	Therefore every $+$-null leaf meets every $-$-null leaf. Together with uniqueness of
	intersection, this gives the claimed bijection.
\end{proof}
\begin{remark}\label{rem:not-bare-perp}
The distinction between $S_R(x)$ and $S_L(x)$ is essential. The bare spacelike
relation $x\perp y$ remembers only that $y$ is spacelike to $x$; it does not
remember on which spacelike side of $x$ the point $y$ lies. The corrected
back-and-forth construction therefore uses the two oriented spacelike relations
below rather than arbitrary finite partial isomorphisms of $(\ll,\perp)$.
\end{remark}

For distinct non-null-related points $x,y$, define
\[
x\perp_R y \quad\Longleftrightarrow\quad y\in S_R(x),
\]
and
\[
x\perp_L y \quad\Longleftrightarrow\quad y\in S_L(x).
\]
Then
\[
x\perp y \quad\Longleftrightarrow\quad x\perp_R y\ \text{or}\ x\perp_L y .
\]

On any null-free subset $E\subset\mathcal M$, define two strict relations $<_+$
and $<_-$ by
\[
x<_+ y
\quad\Longleftrightarrow\quad
x\ll y\ \text{or}\ x\perp_R y,
\]
and
\[
x<_- y
\quad\Longleftrightarrow\quad
x\ll y\ \text{or}\ x\perp_L y.
\]
Equivalently, $<_+$ and $<_-$ are the two orders induced by the two null leaf
spaces.

\begin{lem}[Null orders]\label{lem:null-orders}
Let $E\subset\mathcal M$ be null-free. Then $<_+$ and $<_-$ are strict linear
orders on $E$. Moreover, for $x,y\in E$,
\[
x\ll y
\quad\Longleftrightarrow\quad
x<_+ y\ \text{and}\ x<_- y,
\]
\[
x\perp_R y
\quad\Longleftrightarrow\quad
x<_+ y\ \text{and}\ y<_- x,
\]
and
\[
x\perp_L y
\quad\Longleftrightarrow\quad
x<_- y\ \text{and}\ y<_+ x.
\]
\end{lem}

\begin{proof}
By Lemma~\ref{lem:four-regions}, for any two distinct points $x,y\in E$ exactly
one of the alternatives
\[
x\ll y,
\qquad
y\ll x,
\qquad
x\perp_R y,
\qquad
x\perp_L y
\]
holds. This gives trichotomy for each of $<_+$ and $<_-$.

Transitivity follows from Lemma~\ref{lem:full-null-grid}: after identifying
$\mathcal M$ with the product of the two ordered null leaf spaces, $<_+$ and
$<_-$ are simply the two coordinate orders. The displayed equivalences are the
four-region decomposition written in these two coordinate orders.
\end{proof}

\begin{lem}[Finite null rectangles]\label{lem:finite-null-rectangles}
Let $A\subset\mathcal M$ be finite and null-free. Suppose that finite lower and
upper sets for the two null orders are given:
\[
A_+^-,\ A_+^+\subset A,
\qquad
A_-^-,\ A_-^+\subset A,
\]
with
\[
A_+^-<_+ A_+^+,
\qquad
A_-^-<_- A_-^+,
\]
where empty lower or upper sets are allowed. Then the cell
\[
\Omega=
\{y\in\mathcal M:
A_+^-<_+ y<_+ A_+^+,
\quad
A_-^-<_- y<_- A_-^+\}
\]
is a nonempty open subset of $\mathcal M$.
\end{lem}
\begin{proof}
	For each inequality appearing in the definition of $\Omega$, the corresponding
	condition is open. Indeed, by Lemma~\ref{lem:four-regions}, conditions of the
	form $a<_+y$, $y<_+a$, $a<_-y$, and $y<_-a$ are unions of the appropriate
	chronological or oriented-spacelike components relative to $a$, hence are open.
	Therefore $\Omega$ is open as a finite intersection of open sets.
	
	It remains to prove nonemptiness. The inequalities in the $<_+$-order determine
	a nonempty interval of $+$-null leaves, and the inequalities in the $<_-$-order
	determine a nonempty interval of $-$-null leaves. Choose one $+$-null leaf and one
	$-$-null leaf in these intervals. By Lemma~\ref{lem:full-null-grid}, these two leaves
	meet in a unique point $y\in\mathcal M$. By construction this point satisfies
	all the defining inequalities, so $y\in\Omega$.
\end{proof}

Let $C\subset\mathcal M$ and $D\subset\mathbb M^2$ be countable dense null-free
sets. The target $\mathbb M^2$ is equipped with its standard two null orders,
also denoted $<_+$ and $<_-$.

Let $\mathscr P$ be the set of all finite partial bijections
\[
p:A\longrightarrow B,
\qquad
A\subset C,
\quad B\subset D,
\]
such that, for all $x,y\in A$,
\[
x<_+y \quad\Longleftrightarrow\quad p(x)<_+p(y),
\]
and
\[
x<_-y \quad\Longleftrightarrow\quad p(x)<_-p(y).
\]
Thus elements of $\mathscr P$ preserve both intrinsic null orders.

\begin{lem}[Admissible maps preserve the causal structure]\label{lem:P-preserves-causal}
If $p\in\mathscr P$, then $p$ preserves and reflects $\ll$, $\perp_R$, $\perp_L$,
and hence also $\perp$.
\end{lem}

\begin{proof}
This follows immediately from Lemma~\ref{lem:null-orders}. For example,
\[
x\ll y
\quad\Longleftrightarrow\quad
x<_+y\ \text{and}\ x<_-y,
\]
and both order relations are preserved and reflected by $p$.
\end{proof}

\begin{lem}[Back and forth]\label{lem:P-back-and-forth}
The family $\mathscr P$ is a back-and-forth system.
\end{lem}

\begin{proof}
We prove Forth; Back is identical with source and target interchanged.
Let $p:A\to B$ belong to $\mathscr P$, and let $c\in C\setminus A$. Define the
finite cuts of $c$ over $A$ by
\[
A_+^-:=\{x\in A:x<_+c\},\qquad
A_+^+:=\{x\in A:c<_+x\},
\]
and
\[
A_-^-:=\{x\in A:x<_-c\},\qquad
A_-^+:=\{x\in A:c<_-x\}.
\]
Transport these cuts to $B$ by $p$. Since $p$ preserves both finite orders, the
transported cuts are consistent. In $\mathbb M^2$ they define an open nonempty
null rectangle
\[
\Omega_D=
\{y\in\mathbb M^2:
 p(A_+^-)<_+y<_+p(A_+^+),
\quad
 p(A_-^-)<_-y<_-p(A_-^+)\}.
\]
Since $D$ is dense and $B$ is finite, choose $d\in D\cap\Omega_D\setminus B$.
Then $d$ has the same two null-order cuts over $B$ as $c$ has over $A$, and
$p\cup\{(c,d)\}\in\mathscr P$.

For Back, start with $d\in D\setminus B$, transport its two finite cuts over $B$
back to $A$ using $p^{-1}$, and apply Lemma~\ref{lem:finite-null-rectangles} in
$\mathcal M$. The resulting source cell is open and nonempty, hence contains a
point $c\in C\setminus A$ by density of $C$ and finiteness of $A$. Then
$p\cup\{(c,d)\}\in\mathscr P$.
\end{proof}

\begin{prop}[Dense back-and-forth]\label{prop:corrected-back-and-forth}
Given $a\in C$ and $b\in D$, there exists a bijection
\[
f:C\longrightarrow D
\]
such that $f(a)=b$ and, for all $x,y\in C$,
\[
x<_+y \quad\Longleftrightarrow\quad f(x)<_+f(y),
\qquad
x<_-y \quad\Longleftrightarrow\quad f(x)<_-f(y).
\]
Consequently,
\[
x\ll y \quad\Longleftrightarrow\quad f(x)\ll f(y),
\qquad
x\perp y \quad\Longleftrightarrow\quad f(x)\perp f(y).
\]
\end{prop}

\begin{proof}
The singleton map $\{a\mapsto b\}$ belongs to $\mathscr P$ vacuously. Enumerate
$C$ and $D$, and alternately apply the Forth and Back clauses of
Lemma~\ref{lem:P-back-and-forth}. The union of the resulting increasing chain of
finite partial maps is a bijection $f:C\to D$. Since each finite stage preserves
both null orders, so does $f$. The final claims follow from
Lemma~\ref{lem:null-orders}.
\end{proof}

\subsection{Proof of Theorem~\ref{main} (directed non-compact case)}
\begin{proof}[Proof of Theorem~\ref{main}]
Fix $a\in\mathcal M$ and $b\in\mathbb M^2$. Choose countable dense null-free sets
$C\subset\mathcal M$ and $D\subset\mathbb M^2$ with $a\in C$ and $b\in D$
(Lemma~\ref{null-free}). By Proposition~\ref{prop:corrected-back-and-forth},
there exists a bijection $f:C\to D$ with $f(a)=b$ preserving and reflecting the
two null orders. By Lemma~\ref{lem:P-preserves-causal}, $f$ preserves and
reflects $\ll$. Lemma~\ref{lem:Phi} therefore extends $f$ uniquely to a causal
isomorphism
\[
\Phi:\mathcal M\longrightarrow\mathbb M^2
\]
with $\Phi(a)=b$.

Conjugation by $\Phi$ gives
\[
\Aut(\mathcal M)\cong \Aut(\mathbb M^2).
\]
This proves the theorem.
\end{proof}

\subsection{Two-dimensional spacetimes with compact Cauchy surfaces}
To handle the case where a spacetime $\mathcal{M}$ has compact Cauchy surfaces, we use the following result:

\begin{theorem}[{\cite[Proposition~3.1 and Theorem~3.5]{K}}]\label{thm:normalizer-universal-cover}
Let $\mathcal{M}$ be a Lorentzian manifold (or a semi-Riemannian manifold) with universal covering $\pi: \overline{\mathcal{M}} \to  \mathcal{M}$. Denote by $\Gamma$ the group $\pi_1(\mathcal{M})$.

Then, $\Gamma \subset \Aut(\overline{\mathcal{M}}) =: G$ (respectively $\Gamma \subset  \Conf(\overline{\mathcal{M}}) =:G_1$). Also,
$\Aut(\mathcal{M})$ (or, $\Conf(\mathcal{M})$) is isomorphic to $\mathcal{N}(\Gamma)/\Gamma$ in which $\mathcal{N}(\Gamma)$ is the normalizer of $\Gamma$ in $\Aut(\overline{\mathcal{M}})$ (or, $\Conf(\overline{\mathcal{M}})$).
\end{theorem}

We obtain the following:
\begin{theorem}\label{Kim-2}
Let $\mathcal{M}$ be a two-dimensional globally hyperbolic spacetime having compact Cauchy surfaces. Denote by $\widetilde{\mathcal{M}}$ its universal covering space.

Assuming furthermore that $\widetilde{\mathcal{M}}$ is directed, the group Aut$(\mathcal{M})$ of causal automorphisms of $\mathcal{M}$ is given by
$$ \mathcal{N}(\mathbb{Z}) /\mathbb{Z} $$
where $\mathbb{Z} \cong \pi_1(\mathcal{M})$ is the deck group of the universal cover (acting on $\widetilde{\mathcal{M}}$ by deck transformations), and $\mathcal{N}(\mathbb{Z})$ is the normalizer of $\mathbb{Z}$ in Aut$(\mathbb{M}^2)$.
\end{theorem}
\begin{proof}
	Since $\widetilde{\mathcal{M}}$ is directed with non-compact Cauchy surfaces, Theorem~\ref{main}
	gives $\Aut(\widetilde{\mathcal{M}})\cong \Aut(\mathbb{M}^2)$. Applying
	Theorem~\ref{thm:normalizer-universal-cover} to the deck group
	$\Gamma\cong\pi_1(\mathcal{M})\cong\mathbb{Z}$ yields
	\[
	\Aut(\mathcal{M})\cong \mathcal{N}(\mathbb{Z})/\mathbb{Z},
	\]
	as claimed.
\end{proof}

\subsection{Non-directed two-dimensional globally hyperbolic spacetime}\label{embedding}
Let us define
\[
\mathcal D := \{(x,t)\in \mathbb M^2: |x|+|t|<1\}.
\]
This is the diamond \(I^+((0,-1))\cap I^-((0,+1))\). By
\cite[Theorem~3.5]{K2}, it is causally isomorphic to \(\mathbb M^2\).

\begin{theorem}[{\cite[Theorem~3.6]{K2}}]\label{thm:embed-diamond}
Any two-dimensional spacetime with non-compact Cauchy surfaces can be causally
isomorphically embedded into \(\mathcal D\).
\end{theorem}

We fix such an embedding and denote it by
\[
\iota:\mathcal M\hookrightarrow \mathcal D.
\]
In the normalized form used in Kim's classification of non-compact Cauchy-surface
spacetimes, the image \(\iota(\mathcal M)\) is an open connected subset of
\(\mathcal D\), contains the standard Cauchy interval
\[
I_0=\{(x,0):-1<x<1\},
\]
and satisfies Kim's \(\partial\)-condition. These hypotheses are precisely those
needed to apply Kim's extension theorem for causal isomorphisms between such
subsets of \(\mathcal D\) \cite[Theorem~4.3]{K2}. This point is important: the
extension statement used below is not an assertion about arbitrary embedded
subsets of the diamond, but about the normalized domains covered by Kim's
hypotheses.

\medskip\noindent\emph{Embedding boundary.}
For the remainder of this section we fix the associated \emph{embedding boundary}
\[
  \partial_\iota\mathcal M
  := \overline{\iota(\mathcal M)}\setminus \iota(\mathcal M),
\]
where the closure is taken in the closed diamond
\(\overline{\mathcal D}\subset\mathbb M^2\), and
\(\partial_\iota\mathcal M\) is equipped with the induced subspace topology.
Whenever we write causal relations involving points of \(\partial_\iota\mathcal M\),
they are computed using the ambient Minkowski causal order on
\(\overline{\mathcal D}\).

Let us recall the statement for \(\mathcal M=\mathbb M^2\):
\begin{theorem}[{\cite{GM}}]\label{thm:aut-M2}
The group \(\Aut(\mathbb M^2)\) of causal automorphisms of the Minkowski plane is
\[
(\Homeo_\leq(\mathbb R))^2\rtimes S_2.
\]
\end{theorem}

\medskip\noindent\emph{Convention.}
For \(\Phi\in\Aut(\mathcal M)\), write
\[
\Phi_\iota:=\iota\circ\Phi\circ\iota^{-1}:\iota(\mathcal M)\to\iota(\mathcal M).
\]
Since \(\iota(\mathcal M)\) is one of Kim's normalized domains, Kim's extension
result \cite[Theorem~4.3]{K2} applies to \(\Phi_\iota\) and gives a unique causal
automorphism
\[
\widetilde\Phi_\iota:\mathcal D\to\mathcal D
\]
restricting to \(\Phi_\iota\). By Kim's compactification extension result
\cite[Theorem~5.2]{K2}, this further extends uniquely to a causal automorphism
\[
\widehat\Phi_\iota:\overline{\mathcal D}\to\overline{\mathcal D}.
\]
Whenever we speak of the induced map on \(\partial_\iota\mathcal M\), we mean the
restriction of this map \(\widehat\Phi_\iota\) to the embedding boundary.

\begin{theorem}\label{thm:boundary-induced-homeo}
Let \(\mathcal M\) be a two-dimensional globally hyperbolic spacetime with
non-compact Cauchy surfaces, and fix the normalized Kim embedding
\(\iota:\mathcal M\hookrightarrow\mathcal D\) above. Then every causal
automorphism \(\Phi\in\Aut(\mathcal M)\) induces an order isomorphism
\[
\partial_\iota\mathcal M\to \partial_\iota\mathcal M,
\]
where the order on \(\partial_\iota\mathcal M\) is the ambient Minkowski causal
order on \(\overline{\mathcal D}\). Moreover, the induced map is a homeomorphism
of \(\partial_\iota\mathcal M\) with its subspace topology.
\end{theorem}

\begin{proof}
Let \(\Phi\in\Aut(\mathcal M)\). The conjugate
\(\Phi_\iota=\iota\circ\Phi\circ\iota^{-1}\) is a causal isomorphism from
\(\iota(\mathcal M)\) onto itself. By the preceding convention, Kim's hypotheses
apply: \(\iota(\mathcal M)\) is an open connected subset of \(\mathcal D\), it
contains \(I_0\), and it satisfies the \(\partial\)-condition. Hence
\cite[Theorem~4.3]{K2} gives a unique causal automorphism
\(\widetilde\Phi_\iota:\mathcal D\to\mathcal D\) extending \(\Phi_\iota\). By
\cite[Theorem~5.2]{K2}, this extends uniquely to a causal automorphism
\(\widehat\Phi_\iota:\overline{\mathcal D}\to\overline{\mathcal D}\).

Since \(\widehat\Phi_\iota\) is a homeomorphism and
\(\widehat\Phi_\iota(\iota(\mathcal M))=\iota(\mathcal M)\), it preserves
closures and therefore sends
\[
\partial_\iota\mathcal M
=\overline{\iota(\mathcal M)}\setminus\iota(\mathcal M)
\]
onto itself. Since \(\widehat\Phi_\iota\) is an order isomorphism of the closed
diamond with respect to the ambient causal order, its restriction to
\(\partial_\iota\mathcal M\) is an order isomorphism. The homeomorphism assertion
is the corresponding restriction of the ambient homeomorphism.
\end{proof}

Let \(\Gamma_{\mathcal D}\) denote the ambient causal automorphism group of the
closed diamond, i.e. the restrictions to \(\overline{\mathcal D}\) of the standard
two-null-variable causal automorphisms. Thus, up to the possible exchange of the
two null coordinates,
\[
\Gamma_{\mathcal D}\subseteq
(\Homeo_{\leq}(\mathbb R))^2\rtimes S_2.
\]

\begin{theorem}[Boundary-stabilizer upper bound]\label{thm:nondirected-boundary-upper}
Let \(\mathcal M\) be a two-dimensional globally hyperbolic spacetime with
non-compact Cauchy surfaces, and fix a normalized Kim embedding
\(\iota:\mathcal M\hookrightarrow\mathcal D\). Then there is an injective
homomorphism
\[
\Aut(\mathcal M)
\hookrightarrow
\Stab_{\Gamma_{\mathcal D}}(\partial_\iota\mathcal M).
\]
Thus the non-directed boundary statement is a necessary ambient stabilizer
condition, not a realization theorem for the full stabilizer.
\end{theorem}

\begin{proof}
The map sends \(\Phi\in\Aut(\mathcal M)\) to the ambient extension
\(\widehat\Phi_\iota\). By Theorem~\ref{thm:boundary-induced-homeo}, this ambient
extension preserves \(\partial_\iota\mathcal M\), so its image lies in the stated
stabilizer. Uniqueness of the Kim extension gives injectivity.
\end{proof}

The preceding theorem is the general non-directed statement used below. Further
reductions occur only after specifying additional structure of the embedding
boundary.

\begin{definition}[Graph component]\label{def:graph-component}
A subset \(C^\circ\subset\partial_\iota\mathcal M\) will be called a
nondegenerate graph component if, after possibly exchanging the two null
coordinates \((u,v)\) on \(\mathcal D\), there is a nonempty open interval
\(I\subset\mathbb R\) and a continuous monotone homeomorphism
\[
h:I\to h(I)
\]
such that
\[
C^\circ=\{(u,h(u)):u\in I\},
\]
and \(C^\circ\) is relatively open in its connected boundary component.
\end{definition}

\begin{prop}[Matched boundary action]\label{prop:matched-boundary-action}
Let \(C^\circ\subset\partial_\iota\mathcal M\) be a nondegenerate graph component
as in Definition~\ref{def:graph-component}. Let \(G_{C^\circ}\) be the subgroup
of \(\Aut(\mathcal M)\) whose ambient extensions preserve \(C^\circ\) setwise.
Then the ambient image of \(G_{C^\circ}\) is contained in the group of ambient
maps of the form
\[
(u,v)\longmapsto (f(u),g(v)),
\]
or, after the null-coordinate exchange, the corresponding swapped form, such that
\[
f(I)=I,\qquad g(h(I))=h(I),
\]
and
\[
g\circ h=h\circ f
\qquad\text{on }I.
\]
Thus the induced action on the graph component is of one-reparametrization type:
once \(f\) is known on \(I\), the action of \(g\) on \(h(I)\) is forced by
\[
g=h\circ f\circ h^{-1}.
\]
The actual automorphism group may be a proper subgroup, because the ambient map
must also preserve the remaining boundary data
\(\partial_\iota\mathcal M\setminus C^\circ\).
\end{prop}

\begin{proof}
Let \(\Phi\in G_{C^\circ}\), and let \(\widehat\Phi_\iota\) be its ambient
extension. Up to the possible exchange of the two null coordinates, write
\[
\widehat\Phi_\iota(u,v)=(f(u),g(v)).
\]
Since \(\widehat\Phi_\iota\) preserves \(C^\circ\) setwise, for every \(u\in I\)
we have
\[
(f(u),g(h(u)))\in C^\circ.
\]
But points of \(C^\circ\) are exactly the points of the form \((s,h(s))\) with
\(s\in I\). Hence \(f(I)=I\), \(g(h(I))=h(I)\), and
\[
g(h(u))=h(f(u))
\]
for all \(u\in I\). This is precisely \(g\circ h=h\circ f\). The final assertion
follows because preservation of \(C^\circ\) is only one boundary constraint; the
full embedding boundary must still be preserved.
\end{proof}

\begin{remark}[Boundary stabilizers and graph components]\label{rem:boundary-stabilizers-graphs}
The non-directed case should be understood in two layers. In full generality one
obtains the boundary-stabilizer upper bound
\[
\Aut(\mathcal M)\hookrightarrow
\Stab_{\Gamma_{\mathcal D}}(\partial_\iota\mathcal M).
\]
When a nondegenerate graph component of \(\partial_\iota\mathcal M\) is preserved,
the induced action on that component is matched rather than factorized: the two
null reparametrizations satisfy \(g\circ h=h\circ f\). If no such component is
singled out, or if the residual boundary is singular, no reduction to a one-sector
group is asserted beyond the stabilizer bound. In particular, a reduction to the
stabilizer of a countable set is valid only when the relevant residual boundary
data have actually been proved countable.
\end{remark}

\begin{theorem}\label{main-1}
Let \(\mathcal M\) be a two-dimensional globally hyperbolic spacetime.
\begin{enumerate}[(i)]
\item If \(\mathcal M\) has non-compact Cauchy surfaces and is directed, then
\[
\Aut(\mathcal M)\cong (\Homeo_{\leq}(\mathbb R))^2\rtimes S_2.
\]
\item If \(\mathcal M\) has compact Cauchy surfaces and its universal cover is
directed, then
\[
\Aut(\mathcal M)\cong \mathcal N_{\Gamma}(\mathbb Z)/\mathbb Z,
\qquad
\Gamma=(\Homeo_{\leq}(\mathbb R))^2\rtimes S_2.
\]
\item If \(\mathcal M\) has non-compact Cauchy surfaces and is non-directed, then,
for every fixed normalized Kim embedding \(\iota:\mathcal M\hookrightarrow\mathcal D\),
there is an injective homomorphism
\[
\Aut(\mathcal M)\hookrightarrow
\Stab_{\Gamma_{\mathcal D}}(\partial_\iota\mathcal M).
\]
If a preserved nondegenerate graph component
\(C^\circ=\{(u,h(u)):u\in I\}\subset\partial_\iota\mathcal M\) is present, then on
that component the two null reparametrizations satisfy
\[
g\circ h=h\circ f.
\]
\item If \(\mathcal M\) has compact Cauchy surfaces and universal cover
\(\widetilde{\mathcal M}\) with deck group \(\Delta\), apply the relevant preceding
alternative to \(\widetilde{\mathcal M}\). More precisely, if \(H\) denotes either
the directed model group \(\Gamma\) or the appropriate boundary-stabilizer
upper-bound group for \(\widetilde{\mathcal M}\), then
\[
\Aut(\mathcal M)\hookrightarrow \mathcal N_H(\Delta)/\Delta,
\]
with equality in the directed case described in item~(ii).
\end{enumerate}
All non-directed assertions are upper-bound statements: the ambient stabilizer or
matching condition is necessary for a spacetime automorphism, but no surjectivity
onto the displayed ambient group is asserted.
\end{theorem}

\begin{proof}
Item~(i) is Theorem~\ref{main}. Item~(ii) is Theorem~\ref{Kim-2}. Item~(iii) is
Theorem~\ref{thm:nondirected-boundary-upper}, together with
Proposition~\ref{prop:matched-boundary-action} when the boundary contains a
preserved graph component. Item~(iv) is the same normalizer mechanism applied to
the universal cover: a spacetime automorphism of \(\mathcal M\) lifts to an
automorphism of \(\widetilde{\mathcal M}\) normalizing the deck group, and hence
lands in \(\mathcal N_H(\Delta)/\Delta\), where \(H\) is the relevant directed
model group or non-directed boundary-stabilizer upper-bound group. The final
sentence records the fact that all non-directed arguments above construct
injective homomorphisms into ambient boundary-stabilizer groups; they do not prove
that every ambient stabilizer is realized by an automorphism of the spacetime.
\end{proof}

\subsection{Group of conformal automorphisms of \texorpdfstring{$\mathcal{M}$}{M}}
Let us denote by $\Diff_{\le}(\mathbb{R})$ the group of increasing $C^\infty$-diffeomorphisms of the real line.

In \cite{K}, a characterization of the conformal automorphism group of a two-dimensional globally hyperbolic spacetime with a non-compact Cauchy surface is given.
We obtain a simpler characterization of $C^\infty$-conformal diffeomorphisms of the Minkowski plane using null coordinates $x^+ = t+x$, $x^-= t-x$ (the standard metric on $\mathbb{M}^2$ written in terms of null coordinates is given by $-dx^+dx^-$):

\begin{theorem}\label{main-2}
The group $\Conf(\mathbb{M}^2)$ of conformal diffeomorphisms of the plane is given by
$$      (\Diff_{\le}(\mathbb{R}))^{2} \rtimes D_2, $$
where $D_2$ is the dihedral group, $D_2 \simeq C_2 \times C_2$.
\end{theorem}

\begin{proof}
Let $\Phi : \mathbb{M}^2 \to \mathbb{M}^2$, $(x^+, x^-) \mapsto (X^+,X^-)$ be a conformal diffeomorphism of the plane.

We have:
$$ -dX^+dX^- = -e^{\varphi(x^+,x^-)} dx^+dx^-, $$
as $\Phi$ is a conformal map, with $\varphi(x^+,x^-)$ being an arbitrary function.

Writing $X^+ = X^+(x^+,x^-)$ and $X^- = X^-(x^+,x^-)$ we obtain
$$ -\Big(\frac{\partial X^+}{\partial x^+} dx^+ + \frac{\partial X^+}{\partial x^-} dx^-\Big)\Big(\frac{\partial X^-}{\partial x^+} dx^+ + \frac{\partial X^-}{\partial x^-} dx^-\Big) = -e^{\varphi(x^+,x^-)} dx^+dx^-, $$
hence
\begin{eqnarray*}
 \partial_+ X^+ \partial_+ X^- & = & 0   \\
 \partial_- X^+ \partial_- X^- & =  & 0    \\
 \partial_+ X^+ \partial_- X^- + \partial_- X^+ \partial_+ X^- & > &  0
\end{eqnarray*}
Finally,
\begin{eqnarray*}
  X^+ = X^+(x^+),  X^- = X^-(x^-), & \textrm{and} & X^+_+X^-_- >0  \\
  &  \textrm{Or}   &  \\
  X^+ = X^+(x^-),  X^- = X^-(x^+), & \textrm{and} & X^+_-X^-_+ >0.
\end{eqnarray*}
The rest of the proof now follows \cite{GM}: any conformal diffeomorphism of $\mathbb{M}^2$ is of one of the forms $(x^+,x^-) \mapsto (f(x^+), g(x^-))$ or $(x^+,x^-) \mapsto (f(x^-), g(x^+))$, where $f,g$ are diffeomorphisms of the real line, which are both increasing or decreasing.

It follows that the group $\Conf(\mathbb{M}^2)$ is generated by elements in $(\Diff_{\le}(\mathbb{R}))^{2}$, space reflections $(x^+, x^-) \mapsto (x^-, x^+)$ and time reflections $(x^+,x^-) \mapsto (-x^-, -x^+)$.

We still need to see that the subgroup $(\Diff_{\le}(\mathbb{R}))^{2}$ is normal in $\Conf(\mathbb{M}^2)$. Let $p_1$ denote the transposition $(x^+, x^-) \mapsto (x^-, x^+)$; for any $n \in  (\Diff_{\le}(\mathbb{R}))^{2}$, $n: (x^+,x^-) \mapsto (f(x^+), g(x^-))$ we have $p_1np_1: (x^+,x^-) \mapsto (g(x^+), f(x^-))$ where $g, f$ are increasing diffeomorphisms of $\mathbb{R}$.
Hence $p_1np_1 \in  (\Diff_{\le}(\mathbb{R}))^{2}$ as required.

Let $p_2$ denote the map $(x^+, x^-) \mapsto (-x^-, -x^+)$; then we have $$p_2np_2: (x^+,x^-) \mapsto (-x^-, -x^+) \mapsto (f(-x^-), g(-x^+)) \mapsto (-g(-x^+), -f(-x^-)).  $$
Defining $f_1, g_1$ by $f_1(x^+) = -g(-x^+)$, $g_1(x^-) = -f(-x^-)$ it can be seen that both $f_1$ and $g_1$ are increasing diffeomorphisms  of $\mathbb{R}$. Let now $p := p_1p_2: (x^+,x^-) \mapsto (-x^+,-x^-)$ (note that $p^2=\mathrm{id}$, so the discrete group generated by $p_1,p_2$ is $\{1,p_1,p_2,p_1p_2\}\simeq C_2\times C_2 = D_2$); a similar reasoning shows that $pnp$ is in $(\Diff_{\le}(\mathbb{R}))^{2}$, whence the result follows.

This proves the theorem.
\end{proof}

Applying the above considerations allows us to obtain the following:

\begin{theorem}[Conformal symmetry bounds]\label{thm:2d-aut-class}
Let \(\mathcal M\) be a two-dimensional globally hyperbolic spacetime, and put
\[
\Gamma_1=(\Diff_{\leq}(\mathbb R))^2\rtimes D_2.
\]
\begin{enumerate}[(i)]
\item If \(\mathcal M\) has non-compact Cauchy surfaces and is smoothly conformally
identified with the directed model, then
\[
\Conf(\mathcal M)\cong \Gamma_1.
\]

\item If \(\mathcal M\) has compact Cauchy surfaces and universal cover
\(\widetilde{\mathcal M}\), with deck group \(\Delta\cong\mathbb Z\), then
\[
\Conf(\mathcal M)
\cong
\mathcal N_{\Conf(\widetilde{\mathcal M})}(\Delta)/\Delta.
\]
If \(\Conf(\widetilde{\mathcal M})\) is replaced by an ambient upper-bound subgroup
\(H\), this gives the corresponding upper bound
\[
\Conf(\mathcal M)
\hookrightarrow
\mathcal N_H(\Delta)/\Delta.
\]

\item If \(\mathcal M\) has non-compact Cauchy surfaces and is non-directed, then,
for every fixed normalized Kim embedding \(\iota:\mathcal M\hookrightarrow\mathcal D\),
every conformal automorphism is in particular a causal or anti-causal automorphism
and induces an ambient null-coordinate reparametrization, possibly composed with one
of the discrete null-coordinate symmetries in \(D_2\), preserving the embedding boundary
\(\partial_\iota\mathcal M\). Since \(\iota\) is only a causal isomorphism, this ambient
representative is a priori a homeomorphism, so one obtains the continuous
boundary-stabilizer bound
\[
\Conf(\mathcal M)
\hookrightarrow
\Stab_{(\Homeo_{\le}(\mathbb R))^2\rtimes D_2}(\partial_\iota\mathcal M).
\]
The smooth refinement, with respect to a conformal embedding, is recorded in
Remark~\ref{rem:conformal-smoothness}.

\item If a preserved nondegenerate graph component
\[
C^\circ=\{(u,h(u)):u\in I\}\subset\partial_\iota\mathcal M
\]
is present, then any separated ambient representative
\((u,v)\mapsto(f(u),g(v))\) satisfies the matched boundary condition
\[
g\circ h=h\circ f
\qquad\text{on }I,
\]
with the analogous relation after applying the discrete null-coordinate symmetries.
Since the boundary function \(h\) is only known a priori to be continuous and
monotone, this graph-component statement is an upper-bound constraint unless
additional smoothness of \(h\) is assumed.
\end{enumerate}
Thus, exactly as in the causal case, all non-directed conformal statements are
boundary-stabilizer or graph-component upper bounds; no surjectivity onto the
displayed ambient stabilizer is asserted.
\end{theorem}

\begin{proof}
The first statement is conditional on the stated smooth conformal identification
with the directed model. Under such an identification, conjugation identifies
\(\Conf(\mathcal M)\) with \(\Conf(\mathbb M^2)\), and Theorem~\ref{main-2} gives
\(\Gamma_1\). The compact statement is the standard normalizer description for
quotients by the deck group: a conformal
automorphism of \(\mathcal M\) lifts to a conformal automorphism of
\(\widetilde{\mathcal M}\) normalizing \(\Delta\), and conversely any element of
\(\mathcal N_{\Conf(\widetilde{\mathcal M})}(\Delta)\) descends to the quotient.
If one has only an ambient upper bound \(H\) for \(\Conf(\widetilde{\mathcal M})\),
the same argument gives an embedding into \(\mathcal N_H(\Delta)/\Delta\).

Assume now that \(\mathcal M\) has non-compact Cauchy surfaces and is
non-directed. A conformal diffeomorphism is, up to the discrete exchanges and
reversals described in Theorem~\ref{main-2}, a causal or anti-causal automorphism
preserving the two null directions. Its ambient extension is obtained by the same
Kim extension mechanism used in Theorem~\ref{thm:boundary-induced-homeo}, applied to
the underlying causal or anti-causal map after composing, if necessary, with the
appropriate discrete reflection. This ambient representative is a causal automorphism
of the closed diamond, hence an element of \((\Homeo_{\le}(\mathbb R))^2\rtimes D_2\)
preserving \(\partial_\iota\mathcal M\); this gives the continuous boundary-stabilizer
bound. Smoothness of the ambient representative is not automatic, since \(\iota\) is
only a causal isomorphism; see Remark~\ref{rem:conformal-smoothness}.

If a graph component \(v=h(u)\) is preserved and an ambient map is written in
separated form \((u,v)\mapsto(f(u),g(v))\), preservation of the graph is exactly the
condition \(g\circ h=h\circ f\). Conjugating by the discrete generators in \(D_2\)
gives the corresponding anti-causal or coordinate-swapped versions. Since no
smoothness of \(h\) is asserted in the causal embedding construction, this is a
necessary ambient constraint rather than a realization claim for the full displayed
group.
\end{proof}

\begin{remark}[Smooth refinement of the non-directed conformal bound]\label{rem:conformal-smoothness}
The continuous bound in Theorem~\ref{thm:2d-aut-class}(iii) cannot be upgraded to the
smooth stabilizer $\Stab^{\infty}_{\Gamma_1}(\partial_\iota\mathcal M)$ by the Kim
mechanism alone: in $1+1$ dimensions a causal isomorphism need not be smooth, so
conjugating a (smooth) conformal automorphism by the merely topological embedding
$\iota$ need not yield a smooth ambient map. A smooth description is nonetheless
available through a conformal embedding. Since $\mathcal M\cong\mathbb R^2$ is simply
connected, it is globally conformally flat, hence admits a smooth conformal
diffeomorphism $j:\mathcal M\to\Omega$ onto an open subset $\Omega\subset\mathbb M^2$.
With respect to the associated boundary $\partial_j\mathcal M$, every conformal
automorphism is carried to a smooth separated null reparametrization by
Theorem~\ref{main-2}, so
\[
\Conf(\mathcal M)\hookrightarrow \Stab^{\infty}_{\Gamma_1}(\partial_j\mathcal M).
\]
The price is that this smooth stabilizer is attached to $\partial_j\mathcal M$ rather
than to the causal-embedding boundary $\partial_\iota\mathcal M$ used elsewhere in this
section.
\end{remark}

\section{Applications: boundary matching mechanisms in \texorpdfstring{$1+1$}{1+1} dimensions}
\label{sec:applications}
This section is heuristic and is meant to relate the classification results of Section~4 to symmetry
patterns familiar from asymptotic and boundary analyses in $1+1$ dimensions. We do not attempt to set up a
full scattering framework here. The key observation is that the directed/non-directed dichotomy of the
symmetry description leads to a \emph{factorized} versus \emph{matched} action of large reparametrization groups
on null-type boundary components. We first state this dichotomy abstractly (\S5.1), then develop the
moving-mirror model as a detailed worked example (\S5.2), and finally discuss connections to FLRW cosmology (\S5.3).

\smallskip
\noindent
Throughout this discussion, ``null-type boundaries'' refer to the relevant null components of the
embedding boundary $\partial_\iota \mathcal M=\overline{\iota(\mathcal M)}\setminus\iota(\mathcal M)$ in the closed diamond
$\overline{\mathcal D}$, arising from the causal embedding
$\iota:\mathcal M\hookrightarrow \mathcal D$ fixed in Section~\ref{embedding}. We denote these components
(when they exist) by $\mathscr I^\pm_\iota\subset \partial_\iota\mathcal M$. We use $\partial_\iota\mathcal M$
only as this controlled completion boundary (not as a full Penrose conformal boundary).

\subsection{Directed vs.\ non-directed: factorized vs.\ matched boundary actions}
The classification of Section~4 has a direct physical reading: it controls how symmetries act on the
null-type boundary components of $\mathcal M$, and the dividing line is directedness.

When $\mathcal M$ is non-compact and directed, Theorem~\ref{main} identifies it, up to causal isomorphism,
with $\mathbb M^2$. In null coordinates $(x^+,x^-)$, causal and conformal automorphisms act on the two null
variables separately, up to discrete exchange and reflection. The two null-type boundary components
therefore carry \emph{independent} reparametrizations---one sector acting on $x^+$, the other on
$x^-$---a \emph{factorized} boundary action.

When $\mathcal M$ is non-directed, the embedding of Section~\ref{embedding} may present a null-type boundary
component as a monotone relation between the two null coordinates,
\[
x^- = h(x^+).
\]
A causal automorphism $(x^+,x^-)\mapsto (f(x^+),g(x^-))$ preserves this component only if
\begin{equation}\label{eq:matching}
	g\circ h = h\circ f,
\end{equation}
so the two reparametrizations can no longer be chosen independently: wherever $h$ is invertible,
\eqref{eq:matching} forces $g = h\circ f\circ h^{-1}$. The boundary relation therefore couples the two null
sectors into a single \emph{matched} action, in contrast with the factorized case above
(Figure~\ref{fig:matched-unmatched}). In the smooth category the same mechanism applies with
$\Homeo_{\le}(\mathbb R)$ replaced by $\Diff_{\le}(\mathbb R)$.

In the language of Section~4, this is the passage from the two-sector upper bound
$\Gamma_1=(\Diff_{\le}(\mathbb R))^2\rtimes D_2$ to the one-sector behaviour
$\Gamma_2=\Diff_{\le}(\mathbb R)\rtimes D_2$ of Theorem~\ref{thm:2d-aut-class}, together with the
stabilizer reductions and normalizer-quotient descriptions recorded there.

The simplest instance of \eqref{eq:matching} is an affine matching map, illustrated by AdS$_2$.

\begin{example}[AdS$_2$ as a reflecting cavity]
Anti--de Sitter space AdS$_2$ has a timelike conformal boundary, so it is not globally hyperbolic in the
strict sense of the main theorems---classical and quantum evolution require boundary conditions at infinity
\cite{AvisIshamStorey,IshibashiWaldAdS}---but it exhibits the matched mechanism in a transparent form. On the
universal cover, AdS$_2$ is conformal to the strip
$ds^2 \propto -d\tau^2 + d\rho^2$, $\rho\in(-\pi/2,\pi/2)$.
In null variables $x^\pm:=\tau\pm\rho$, the right boundary $\rho=\pi/2$ is the affine relation
\[
x^- = h(x^+) := x^+ - \pi,
\]
and a reparametrization pair $(f,g)$ preserves it exactly when $g\circ h = h\circ f$, i.e.\ \eqref{eq:matching}.
Because $h$ is affine, the matched sector is a rigid conjugate of the $x^+$ reparametrizations; the moving
mirrors of the next subsection realize the same mechanism with a genuinely non-affine matching map.
\end{example}

\begin{figure}[!h]
\begin{center}
\IfFileExists{matched-unmatched.png}{%
	\includegraphics[width=0.9\linewidth]{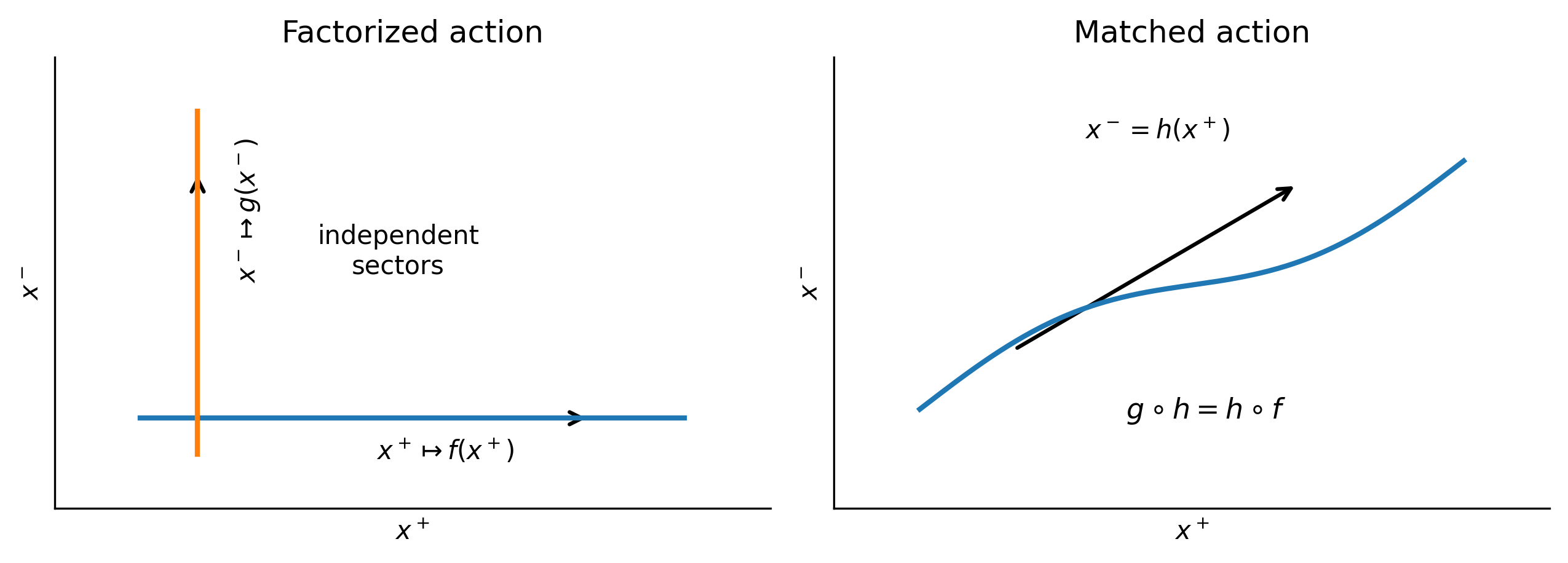}}%
{\fbox{matched-unmatched.png}}
\caption{Schematic of the factorized versus matched boundary action in null coordinates. In the
	directed case, the two null-type boundary components carry independent reparametrizations
	$x^+\mapsto f(x^+)$ and $x^-\mapsto g(x^-)$. In the presence of a boundary component given by a
	monotone graph $x^- = h(x^+)$, boundary preservation forces the matching condition $g\circ h = h\circ f$
	(cf.\ \eqref{eq:matching}).}
\label{fig:matched-unmatched}
\end{center}
\end{figure}

\subsection{Moving mirrors: boundary matching, horizons, and particle creation}
\label{subsec:moving-mirror}
A particularly transparent $1+1$-dimensional application of the ``matched'' symmetry mechanism is provided
by the standard moving-mirror model (see e.g.\ \cite{FullingDavies,CarlitzWilley}). Work in null
coordinates on $\mathbb{M}^2$,
\[
u:=t-x,\qquad v:=t+x,
\]
so that future-directed causality corresponds to the product order $(u,v) \prec (u',v')$ iff
$u\le u'$ and $v\le v'$. A (perfectly reflecting) mirror is specified by a timelike trajectory $x=z(t)$,
which in null coordinates can be written as a graph relation
\begin{equation}\label{eq:mirror-raytracing}
	v=p(u),
\end{equation}
where $p$ is strictly increasing (timelikeness forces monotonicity). Equivalently, one may write $u=h(v)$
with $h:=p^{-1}$ on the relevant range; see Figure~\ref{fig:moving-mirror}.

\begin{figure}[!h]
\begin{center}
\IfFileExists{fig_moving_mirror_horizon_redshift.pdf}{%
	\includegraphics[width=0.92\linewidth]{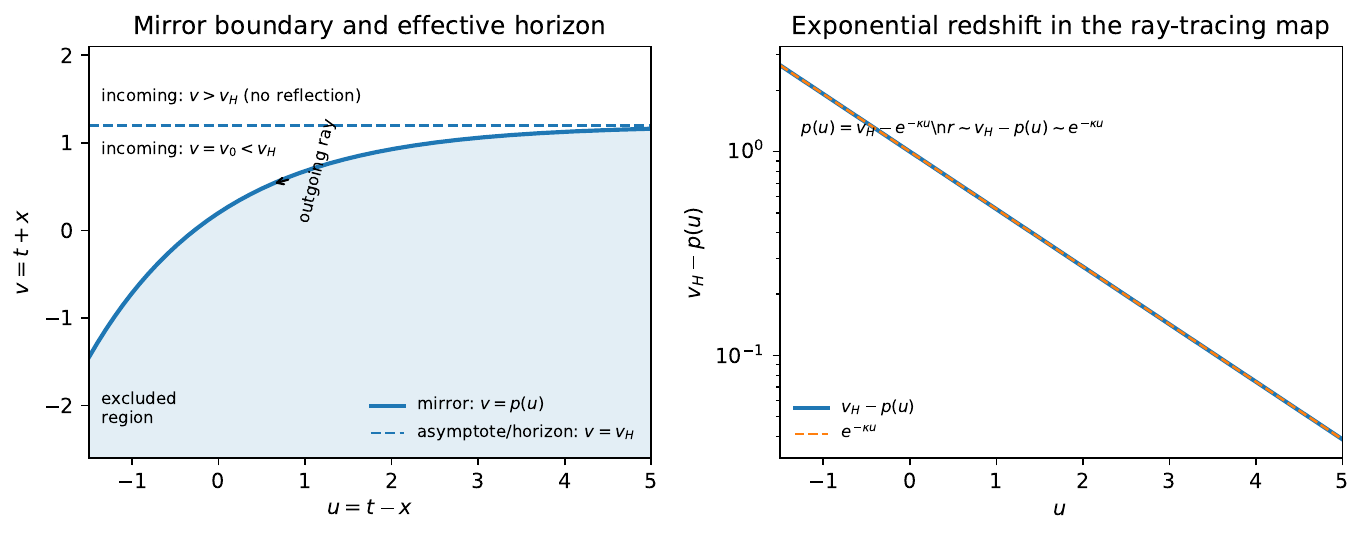}}%
{\fbox{fig\_moving\_mirror\_horizon\_redshift.pdf}}
\caption{A moving mirror induces a ray-tracing map $v=p(u)$ that glues incoming and outgoing null data.
	The boundary relation enforces a matching (conjugacy) constraint on admissible reparametrizations.}
\label{fig:moving-mirror}
\end{center}
\end{figure}

\medskip
\noindent\emph{Mirror as a boundary component in the sense of $\partial_\iota\mathcal{M}$.}
Define the physical region $\mathcal{M}$ to be the open subspacetime on (say) one side of the mirror, so
that the mirror worldline itself is \emph{not} included in $\mathcal{M}$, but appears as a boundary
component of its completion. After choosing a causal embedding $\iota:\mathcal{M}\hookrightarrow\mathcal{D}$
as in Section~\ref{embedding}, the image of the mirror worldline becomes a connected component of the
embedding boundary $\partial_\iota\mathcal{M}$. Thus, in this model the ``mirror'' is literally a component
of $\partial_\iota\mathcal{M}$ (a boundary in the completion sense used throughout the paper), rather than
an asymptotic component at null infinity.

\medskip
\noindent\emph{How the boundary glues the two null sectors.}
For a massless scalar, solutions of $\Box\phi=0$ on $\mathbb{M}^2$ decompose into independent chiral pieces,
\[
\phi(u,v)=F(u)+G(v).
\]
Imposing a reflecting boundary condition on the mirror (e.g.\ Dirichlet) means
\[
0=\phi\bigl(u,p(u)\bigr)=F(u)+G\bigl(p(u)\bigr),
\]
hence
\begin{equation}\label{eq:mirror-gluing}
	G(v)=-F\bigl(p^{-1}(v)\bigr)=-F\bigl(h(v)\bigr).
\end{equation}
In particular, the boundary removes the independent $v$-sector datum: once $F$ is fixed, $G$ is determined
by the matching map $h=p^{-1}$.

\medskip
\noindent\emph{Group-theoretic matching.}
A conformal transformation $(u,v)\mapsto (g(u),f(v))$ preserves the mirror worldline \eqref{eq:mirror-raytracing} if and only if
\begin{equation}\label{eq:mirror-matching}
	f\circ p = p\circ g,
	\qquad\text{equivalently}\qquad
	g\circ h = h\circ f,
\end{equation}
which is the instance of the matching condition \eqref{eq:matching} for this geometry (with $x^-=u$, $x^+=v$).
Consequently, the symmetry group of the spacetime with the reflecting boundary is reduced from a factorized
$\Diff_{\le}(\mathbb{R})\times \Diff_{\le}(\mathbb{R})$ to the graph subgroup cut out by \eqref{eq:mirror-matching}.

\begin{prop}[Matched reparametrization charges for a Dirichlet mirror]\label{prop:mirror-matched-charges}
Let $(\mathbb{M}^2,ds^2)$ be $1{+}1$ Minkowski space with null coordinates
\[
u:=t-x,\qquad v:=t+x,\qquad ds^2=-\,du\,dv.
\]
Fix a smooth strictly increasing function $p:\mathbb{R}\to\mathbb{R}$ and define the moving-mirror spacetime
\[
\mathcal{M}_p:=\{(u,v)\in\mathbb{R}^2 \mid v\ge p(u)\},
\qquad
\gamma:=\partial\mathcal{M}_p=\{v=p(u)\}.
\]
Consider the classical massless scalar field $\phi$ on $\mathcal{M}_p$ satisfying the Dirichlet boundary condition
\[
\phi|_{\gamma}=0,
\]
and let $T_{ab}$ be the (improved or unimproved) classical stress tensor; in null coordinates one has
$T_{uu}=(\partial_u\phi)^2$ and $T_{vv}=(\partial_v\phi)^2$.

Assume $\phi$ has sufficient decay (e.g., $F'\in L^2(\mathbb{R})$, which holds for finite-energy solutions) so that the following integrals converge, and let $\xi\in C_c^\infty(\mathbb{R})$.
Define the outgoing and incoming chiral charges
\[
Q^+[\xi]\ :=\ \int_{\mathbb{R}} du\,\xi(u)\,T_{uu}(u)\quad\text{on }\mathscr{I}^+_\iota,
\qquad
Q^-[\eta]\ :=\ \int_{\mathbb{R}} dv\,\eta(v)\,T_{vv}(v)\quad\text{on }\mathscr{I}^-_\iota.
\]
Then the mirror ray-tracing relation $v=p(u)$ implies the identity
\begin{equation}\label{eq:Qmatch}
	Q^+[\xi]\ =\ Q^-[\,p_*\xi\,],
	\qquad
	(p_*\xi)(v)\ :=\ \xi(h(v))\,p'(h(v)),
	\qquad h:=p^{-1}.
\end{equation}
Equivalently, the charge associated to the vector field $\xi(u)\partial_u$ on $\mathscr{I}^+_\iota$
coincides with the charge associated to the pushed-forward vector field $(p_*\xi)(v)\partial_v$ on
$\mathscr{I}^-_\iota$.
\end{prop}

\begin{proof}
A general smooth solution of $\Box\phi=0$ on $\mathcal{M}_p$ has the form
\[
\phi(u,v)=F(u)+G(v).
\]
The Dirichlet condition $\phi(u,p(u))=0$ gives $F(u)+G(p(u))=0$, hence (setting $h=p^{-1}$)
\[
G(v)=-F(h(v)).
\]
In particular, on $\mathscr{I}^+_\iota$ one may view the outgoing component as $F(u)$, while on
$\mathscr{I}^-_\iota$ one may view the incoming component as $G(v)$, related by $F(u)=-G(p(u))$.

Differentiating yields $F'(u)=-G'(p(u))\,p'(u)$, hence
\[
T_{uu}(u)=(F'(u))^2=(p'(u))^2\,\bigl(G'(p(u))\bigr)^2=(p'(u))^2\,T_{vv}(p(u)).
\]
Therefore,
\[
Q^+[\xi]=\int du\,\xi(u)\,(p'(u))^2 T_{vv}(p(u)).
\]
Changing variables $v=p(u)$ (so $dv=p'(u)\,du$) gives
\[
Q^+[\xi]=\int dv\,\xi(h(v))\,p'(h(v))\,T_{vv}(v)=Q^-[p_*\xi],
\]
which is \eqref{eq:Qmatch}.
\end{proof}

\begin{remark}[Infinitesimal matching and the $g\circ h=h\circ f$ mechanism]\label{rem:mirror-infinitesimal}
At the group level, a separated conformal transformation $(u,v)\mapsto (g(u),f(v))$ preserves the mirror
$v=p(u)$ iff $f\circ p=p\circ g$.

At the Lie-algebra level, a vector field $X=\xi(u)\partial_u+\eta(v)\partial_v$ preserves the mirror iff
\[
\eta(p(u))=p'(u)\,\xi(u),
\]
i.e.\ $\eta=p_*\xi$. Thus the residual symmetry is a matched graph-component subalgebra, in precise parallel with
the matching constraint \eqref{eq:mirror-matching}.
\end{remark}

\medskip
\noindent\emph{Energy flux and particle creation (heuristic reminder).}
In two-dimensional CFT language one may package the resulting energy flux in terms of the Schwarzian
derivative. For a theory of central charge $c$ (with $c=1$ for a free scalar) one has, for the Minkowski
``in'' vacuum,
\begin{equation}\label{eq:mirror-flux-schwarzian}
	\bigl\langle T_{uu}(u)\bigr\rangle_{\mathscr{I}^+_\iota}
	= -\frac{c}{24\pi}\,\{p(u),u\},
	\qquad
	\{p(u),u\}:=\frac{p'''(u)}{p'(u)}-\frac{3}{2}\Bigl(\frac{p''(u)}{p'(u)}\Bigr)^2,
\end{equation}
so that $\{p,u\}=0$ (affine $p$) implies vanishing flux, while genuinely nonlinear $p$ yields nontrivial
radiation.

\begin{remark}[Conformal interfaces and defect lines]\label{rem:interfaces}
The moving-mirror construction is a special case of a more general principle: \emph{any}
distinguished timelike ``defect'' worldline $\gamma\subset\mathbb{M}^2$ that is treated as a boundary
component enforces the same matching constraint. In null coordinates the defect can be described by
a monotone relation $v=p(u)$ (a ray-tracing / gluing map), and preserving it forces $f\circ p = p\circ g$,
i.e.\ exactly the conjugacy constraint~\eqref{eq:matching}. This covers reflecting boundaries
(mirrors), internal interfaces, and impurity lines where left- and right-movers are glued by a
prescribed relation. In many interface problems one may ``unfold'' (folding trick) so that the interface
becomes a boundary condition for a doubled theory on a half-space, converting the defect into a boundary
component in the completion sense used throughout this paper. From this viewpoint, the directed/non-directed
dichotomy (factorized vs.\ matched reparametrizations) is a geometric avatar of the familiar statement
that conformal interfaces glue chiral sectors.
\end{remark}

\subsection{Cosmological toy models: FLRW horizons as an order-theoretic obstruction to directedness}
Many widely used cosmological geometries are of Friedmann--Lema\^{\i}tre--Robertson--Walker type.
In $1+1$ dimensions one may consider
\[
ds^{2} = -dt^{2} + a^{2}(t)\,d\sigma^{2},
\qquad \sigma\in \mathbb{R}\ \text{or}\ S^{1},
\]
with scale factor $a(t)>0$. Introducing conformal time $\eta$ via $d\eta = dt/a(t)$, the metric becomes
\[
ds^{2} = a^{2}(\eta)\,(-d\eta^{2} + d\sigma^{2}),
\qquad (\eta,\sigma)\in I\times\Sigma,
\]
where $\Sigma=\mathbb{R}$ or $S^{1}$. Since causal relations are conformally invariant, the causal order
is determined by the conformal domain $I\times\Sigma$ and is independent of the scale factor $a$.

\begin{lem}[Causal order in conformal coordinates]\label{lem:FRW-order}
Let $p=(\eta,\sigma)$ and $q=(\eta',\sigma')$ in $I\times\Sigma$ with conformal metric $-d\eta^{2}+d\sigma^{2}$.
Then
\[
p\prec q \iff \eta'\ge \eta\ \text{and}\ d_{\Sigma}(\sigma,\sigma')\le \eta'-\eta,
\]
and similarly $p\ll q$ iff $\eta'>\eta$ and $d_{\Sigma}(\sigma,\sigma')<\eta'-\eta$.
\end{lem}

\begin{proof}
The conformal metric $-d\eta^2+d\sigma^2$ is locally isometric to $1{+}1$ Minkowski space; future-directed
causal curves satisfy $|d\sigma/d\eta|\le 1$, so $p\prec q$ requires $\eta'\ge\eta$ and the accumulated
spatial displacement $d_\Sigma(\sigma,\sigma')$ to be at most $\eta'-\eta$ (the light-travel budget).
The chronological relation $p\ll q$ corresponds to strict inequalities (the curve can be taken timelike).
\end{proof}

\medskip
\noindent The following proposition shows that finite conformal-time endpoints obstruct directedness.

\begin{prop}\label{prop:FRW-directedness}
Assume $\operatorname{diam}(\Sigma)>0$ (in particular for $\Sigma=\mathbb{R}$ or $S^{1}$).
Write $I=(\eta_{-},\eta_{+})$ with $\eta_{\pm}\in\mathbb{R}\cup\{\pm\infty\}$.
\begin{enumerate}
	\item If $\eta_{+}<+\infty$, then $I\times\Sigma$ is not upward-directed.
	\item If $\eta_{-}>-\infty$, then $I\times\Sigma$ is not downward-directed.
\end{enumerate}
Consequently, the model is directed (upward and downward) only if $I$ is unbounded to the future and to the past.
\end{prop}

\begin{proof}
We prove (1); (2) is analogous. Assume $\eta_{+}<\infty$ and choose $\eta_{0}\in I$ so that
$2(\eta_{+}-\eta_{0})<\operatorname{diam}(\Sigma)$. Choose $\sigma_{0},\sigma_{1}\in\Sigma$ with
$d_{\Sigma}(\sigma_{0},\sigma_{1})>2(\eta_{+}-\eta_{0})$ and set $p=(\eta_{0},\sigma_{0})$ and
$q=(\eta_{0},\sigma_{1})$. If $r=(\eta_{r},\sigma_{r})$ were a common future of $p$ and $q$,
Lemma~\ref{lem:FRW-order} would give
$d_{\Sigma}(\sigma_{0},\sigma_{r})\le \eta_{r}-\eta_{0}$ and $d_{\Sigma}(\sigma_{1},\sigma_{r})\le \eta_{r}-\eta_{0}$.
By the triangle inequality,
\[
d_{\Sigma}(\sigma_{0},\sigma_{1})\le d_{\Sigma}(\sigma_{0},\sigma_{r}) + d_{\Sigma}(\sigma_{r},\sigma_{1})
\le 2(\eta_{r}-\eta_{0})\le 2(\eta_{+}-\eta_{0}),
\]
contradicting the choice of $\sigma_{0},\sigma_{1}$. Hence no such $r$ exists, so the spacetime is not upward-directed.
\end{proof}

In the non-compact case $\Sigma=\mathbb{R}$, if $\eta$ ranges over all of $\mathbb{R}$ (no conformal-time endpoints),
then the model is directed. If in addition the spacetime is globally hyperbolic (as in the standard FLRW setting),
Theorem~\ref{main} implies that it is causally isomorphic to $\mathbb{M}^{2}$; hence in this regime the induced
boundary symmetry pattern is the factorized one. If instead $\eta$ has a finite past or future endpoint (as in
big-bang/big-crunch type models, or finite conformal-time future as in inflationary horizons), then directedness fails
by Proposition~\ref{prop:FRW-directedness}, and the allowable boundary symmetry action is reduced accordingly, through graph-component matching or the general boundary-stabilizer upper bounds of Theorems~\ref{main-1} and~\ref{thm:2d-aut-class}.

\begin{example}[$1{+}1$ FLRW models with finite conformal time]\label{ex:FLRW-nondirected}
For any $1{+}1$ FLRW model with $\Sigma=\mathbb{R}$ and a finite conformal-time endpoint (e.g.\ the
$1{+}1$ de Sitter slab with $I=(-\pi/2,\pi/2)$), Proposition~\ref{prop:FRW-directedness} shows that
directedness fails, placing the spacetime under the non-directed boundary-stabilizer alternatives of Theorems~\ref{main-1}
and~\ref{thm:2d-aut-class}. In null coordinates $x^{\pm}=\eta\pm\sigma$, the conformal-time boundaries
$\eta=\mathrm{const}$ become lines $x^{-}=-x^{+}+c$ (decreasing functions in null coordinates, i.e.\
spacelike segments of the embedding boundary $\partial_\iota\mathcal{M}$). The precise identification
of which case applies---and the explicit automorphism group---depends on the detailed structure of
$\partial_\iota\mathcal{M}$ for each model and is deferred to future work.
\end{example}

These examples illustrate the main physical meaning of the symmetry descriptions obtained in Section~4: directed spacetimes give rise to factorized null reparametrization sectors, while finite conformal-time endpoints or boundary/defect data force matched or reduced symmetry actions. Further variants arise by prescribing additional boundary data or by iterating the matching mechanism along chains of null interfaces, but we do not pursue these directions here.

\section{Examples}  \label{Examples}
\begin{prop}[Einstein static cylinder]\label{prop:Einstein-cylinder-groups}
	Let
	\[
	\mathrm{Cyl}:=\mathbb R_t\times S^1_\theta,
	\qquad
	g=-dt^2+d\theta^2,
	\]
	with universal cover $\pi:\mathbb M^2\to \mathrm{Cyl}$ given by
	\[
	\pi(t,x)=(t,[x]).
	\]
	In null coordinates $x^\pm=t\pm x$ on $\mathbb M^2$, the deck generator is
	\[
	\gamma(x^+,x^-)=(x^+ + 2\pi,\; x^- - 2\pi).
	\]
	Set
	\[
	\widetilde{\Homeo}_+(S^1)
	:=\{f\in \Homeo_{\le}(\mathbb R): f(u+2\pi)=f(u)+2\pi\},
	\]
	\[
	\widetilde{\Diff}_+(S^1)
	:=\{f\in \Diff_{\le}(\mathbb R): f(u+2\pi)=f(u)+2\pi\},
	\]
	and let $\delta(u)=u+2\pi$.
	Then
	\[
	\Aut(\mathrm{Cyl})
	\cong
	\Bigl(\widetilde{\Homeo}_+(S^1)\times \widetilde{\Homeo}_+(S^1)\Bigr)
	/\langle (\delta,\delta^{-1})\rangle
	\rtimes S_2,
	\]
	and
	\[
	\Conf(\mathrm{Cyl})
	\cong
	\Bigl(\widetilde{\Diff}_+(S^1)\times \widetilde{\Diff}_+(S^1)\Bigr)
	/\langle (\delta,\delta^{-1})\rangle
	\rtimes D_2.
	\]
\end{prop}

\begin{proof}
	By Theorems~\ref{Kim-2} and~\ref{thm:2d-aut-class}, it suffices to compute the normalizers
	of the deck group $\langle\gamma\rangle$ in $\Aut(\mathbb M^2)$ and $\Conf(\mathbb M^2)$.
	
	For the causal group, every element is either of the form
	\[
	(x^+,x^-)\mapsto (f(x^+),g(x^-))
	\]
	or
	\[
	(x^+,x^-)\mapsto (g(x^-),f(x^+)),
	\]
	with $f,g\in\Homeo_{\le}(\mathbb R)$.
	If $\Phi(x^+,x^-)=(f(x^+),g(x^-))$, then
	\[
	\Phi\gamma\Phi^{-1}(u,v)
	=
	\bigl(f(f^{-1}(u)+2\pi),\, g(g^{-1}(v)-2\pi)\bigr).
	\]
	Thus $\Phi$ normalizes $\langle\gamma\rangle$ iff
	\[
	f(s+2\pi)=f(s)+2\pi,
	\qquad
	g(s+2\pi)=g(s)+2\pi,
	\]
	i.e.\ iff $f,g\in\widetilde{\Homeo}_+(S^1)$.
	The null-coordinate swap conjugates $\gamma$ to $\gamma^{-1}$, so
	\[
	N_{\Aut(\mathbb M^2)}(\langle\gamma\rangle)
	=
	\bigl(\widetilde{\Homeo}_+(S^1)\times \widetilde{\Homeo}_+(S^1)\bigr)\rtimes S_2.
	\]
	Since $\langle\gamma\rangle=\langle(\delta,\delta^{-1})\rangle$, quotienting gives the formula for
	$\Aut(\mathrm{Cyl})$.
	
	The conformal case is identical, replacing $\Homeo_{\le}(\mathbb R)$ by $\Diff_{\le}(\mathbb R)$,
	and using Theorem~\ref{main-2} to include the discrete $D_2$-symmetries. Hence
	\[
	N_{\Conf(\mathbb M^2)}(\langle\gamma\rangle)
	=
	\bigl(\widetilde{\Diff}_+(S^1)\times \widetilde{\Diff}_+(S^1)\bigr)\rtimes D_2,
	\]
	and quotienting by $\langle(\delta,\delta^{-1})\rangle$ yields the formula for
	$\Conf(\mathrm{Cyl})$.
\end{proof}
\begin{prop}[Global de Sitter cylinder]\label{prop:dS2-groups}
	Let
	\[
	dS_2:=\Bigl((-\tfrac{\pi}{2},\tfrac{\pi}{2})_\eta\times S^1_\theta,\;
	g=\sec^2\eta\,(-d\eta^2+d\theta^2)\Bigr),
	\]
	where $\theta$ is $2\pi$-periodic. Then
	\[
	\Aut(dS_2)\cong \Homeo_+(S^1)\rtimes S_2,
	\qquad
	\Conf(dS_2)\cong \Diff_+(S^1)\rtimes D_2.
	\]
\end{prop}

\begin{proof}
	Let $\widetilde{dS}_2$ denote the universal cover, with coordinates
	\[
	(\eta,x)\in \Bigl(-\tfrac{\pi}{2},\tfrac{\pi}{2}\Bigr)\times \mathbb R,
	\qquad
	ds^2=\sec^2\eta\,(-d\eta^2+dx^2).
	\]
	In null coordinates
	\[
	u:=\eta+x,\qquad v:=\eta-x,
	\]
	one has
	\[
	ds^2=-\sec^2\!\Bigl(\frac{u+v}{2}\Bigr)\,du\,dv,
	\]
	and the domain is the strip
	\[
	|u+v|<\pi.
	\]
	The deck generator of the covering $\widetilde{dS}_2\to dS_2$ is
	\[
	\tau(u,v)=(u+2\pi,\;v-2\pi).
	\]
	
	\smallskip
	\noindent\emph{Conformal group.}
	Let $\Phi$ be an orientation-preserving conformal diffeomorphism of $\widetilde{dS}_2$.
	Since the metric is conformal to the Minkowski metric on the strip, $\Phi$ has separated form
	\[
	\Phi(u,v)=(f(u),g(v)),
	\qquad
	f,g\in \Diff_{\le}(\mathbb R).
	\]
	The two boundary components of the strip are
	\[
	B_+:\ v=-u+\pi,
	\qquad
	B_-:\ v=-u-\pi.
	\]
	Preserving $B_+$ gives
	\[
	g(-u+\pi)=-f(u)+\pi,
	\]
	while preserving $B_-$ gives
	\[
	g(-u-\pi)=-f(u)-\pi.
	\]
	Comparing these two identities yields
	\[
	f(u+2\pi)=f(u)+2\pi.
	\]
	Conversely, every $f\in\widetilde{\Diff}_+(S^1)$ determines a unique $g$ by
	\[
	g(v)=\pi-f(\pi-v),
	\]
	and the periodicity condition on $f$ ensures that both $B_+$ and $B_-$ are preserved.
	Thus the connected orientation-preserving conformal lift group is naturally identified with
	\[
	\widetilde{\Diff}_+(S^1)
	:=\{f\in\Diff_{\le}(\mathbb R): f(u+2\pi)=f(u)+2\pi\}.
	\]
	
	The deck generator $\tau$ corresponds to the element $\delta(u)=u+2\pi$, so quotienting by
	$\langle\tau\rangle$ gives
	\[
	\Conf_0(dS_2)\cong \widetilde{\Diff}_+(S^1)/\langle\delta\rangle
	\cong \Diff_+(S^1).
	\]
	In addition, the involutions
	\[
	s(u,v)=(v,u),
	\qquad
	t(u,v)=(-v,-u)
	\]
	preserve the strip and normalize $\langle\tau\rangle$; together with their product they generate
	a copy of $D_2$. Hence
	\[
	\Conf(dS_2)\cong \Diff_+(S^1)\rtimes D_2.
	\]
	
	\smallskip
	\noindent\emph{Causal group.}
	The same argument applies to causal automorphisms, replacing
	$\Diff_{\le}(\mathbb R)$ by $\Homeo_{\le}(\mathbb R)$. Thus the connected causal lift group is
	\[
	\widetilde{\Homeo}_+(S^1)
	:=\{f\in\Homeo_{\le}(\mathbb R): f(u+2\pi)=f(u)+2\pi\},
	\]
	and descending to the quotient yields
	\[
	\Aut_0(dS_2)\cong \Homeo_+(S^1).
	\]
	Among the above discrete symmetries, only the null-coordinate swap $s(u,v)=(v,u)$ is causal;
	the time-reversing involution $t$ is anti-causal. Therefore
	\[
	\Aut(dS_2)\cong \Homeo_+(S^1)\rtimes S_2.
	\]
\end{proof}
\section{Conclusion} \label{conc}

We made the order-theoretic back-and-forth mechanism for causal isomorphisms explicit in the
$1+1$-dimensional setting by refining the finite language with the two intrinsic null orders. This yields a
partial classification of causal and conformal automorphism groups of two-dimensional globally
hyperbolic spacetimes, including the directed/non-directed and compact/non-compact cases. In the
directed non-compact case one recovers the Minkowski plane up to causal isomorphism; in the
remaining cases the automorphism groups are described by stabilizer and normalizer-quotient
constructions.

On the physics side, the symmetry description can be read as a factorized-versus-matched alternative for
large reparametrization actions on null-type completion boundaries. Directed spacetimes give
independent actions on the two null directions, whereas non-directed spacetimes force graph-type
matching relations or stabilizer reductions. The moving-mirror, interface, and FLRW examples show that this is a common geometric mechanism rather than a collection of unrelated illustrations. The compact normalizer-quotient cases remain partially implicit in general. Section~\ref{Examples}
shows, however, that in the cases of the Einstein static cylinder and global de Sitter space these quotients can be computed explicitly, yielding circle-reparametrization-type groups in the compact setting. Extending such explicit computations to further compact examples, and extending these results beyond $1+1$ dimensions, will likely require additional
structure beyond the binary causal relations used here.

\section*{Acknowledgements}
I would like to thank an anonymous referee for constructive comments that helped improve the manuscript.
During the revision of this paper, the author used generative AI tools, including Claude and ChatGPT, for editorial assistance, language polishing, organizational restructuring, and expository checks. All mathematical statements, proofs, and final revisions were checked by the author, who takes full responsibility for the content of the paper.
\appendix
\section{Two-dimensional slice geometry}\label{app:slice-geometry}

Throughout this appendix, $(\mathcal M,g)$ denotes a connected, time-oriented, $2$-dimensional
globally hyperbolic spacetime with non-compact Cauchy surfaces. In particular, $\mathcal M$ is a smooth
Lorentzian manifold without boundary.

\subsection{Global splitting and a spatial coordinate}

By global hyperbolicity, there exists a smooth Cauchy temporal function $t:\mathcal M\to\mathbb R$
whose level sets $\Sigma_s:=t^{-1}(s)$ are smooth spacelike Cauchy hypersurfaces, and a diffeomorphism
\[
\Phi:\mathcal M \longrightarrow \mathbb R \times \Sigma_0,\qquad
\Phi(p) = (t(p),\pi(p)),
\]
where $\pi:\mathcal M\to\Sigma_0$ is the projection onto the reference Cauchy surface $\Sigma_0:=\Sigma_{0}$
(see \cite{BernalSanchez,BernalSanchez2}).
Since $\Sigma_0$ is a connected non-compact $1$-manifold, fix once and for all a diffeomorphism
$\chi:\Sigma_0\to\mathbb R$ and define the spatial coordinate
\[
x:\mathcal M\to\mathbb R,\qquad x(p):=\chi(\pi(p)).
\]
Thus $\Phi$ identifies $\mathcal M$ with global coordinates $(t,x)\in\mathbb R^2$, and for each $s$,
the restriction $\pi|_{\Sigma_s}:\Sigma_s\to\Sigma_0$ is a diffeomorphism, so $x|_{\Sigma_s}$ identifies
each slice $\Sigma_s$ with $\mathbb R$ and induces an order on $\Sigma_s$:
for $p,q\in\Sigma_s$, write $p<q$ iff $x(p)<x(q)$.

\subsection{Causal curves meet Cauchy slices exactly once}

\begin{lem}\label{lem:temporal-monotone}
If $\gamma$ is a future-directed causal curve, then $t\circ\gamma$ is strictly increasing.
\end{lem}

\begin{proof}
Since $t$ is temporal, $\nabla t$ is timelike. Equivalently, $dt(v)>0$ for every nonzero future-directed
causal vector $v$. Along a future-directed causal curve $\gamma$, we have $(t\circ\gamma)' = dt(\dot\gamma)>0$.
\end{proof}

\begin{lem}\label{lem:meets-slice-once}
Let $\gamma:I\to\mathcal M$ be an inextendible future-directed causal curve. Then for every $s\in\mathbb R$,
$\gamma$ meets $\Sigma_s$ in exactly one point.
\end{lem}

\begin{proof}
Uniqueness: by Lemma~\ref{lem:temporal-monotone}, $t\circ\gamma$ is strictly increasing, so it attains the value
$s$ at most once.

Existence: since $\Sigma_s$ is a Cauchy surface, every inextendible causal curve meets $\Sigma_s$ (by definition
of Cauchy surface / Cauchy temporal function).
\end{proof}

\subsection{A small-time localization estimate}

We will use the following standard estimate: future-directed causal curves cannot travel far (in an auxiliary
Riemannian metric) when the temporal time $t$ increases only a little.

\begin{lem}\label{lem:small-time-localization}
Fix a point $p\in\mathcal M$ and an open neighborhood $U\ni p$. Then there exists $\varepsilon>0$ such that
\[
J^+(p)\cap t^{-1}\bigl([t(p),\,t(p)+\varepsilon]\bigr)\subset U.
\]
\end{lem}

\begin{proof}
Choose a smooth Riemannian metric $h$ on $\mathcal M$. Pick an open neighborhood $W\ni p$ with compact closure
$\overline W\subset U$. Set $d:=\mathrm{dist}_h(p,\mathcal M\setminus W)>0$.

For each $q\in \overline W$, let
\[
K_q := \{v\in T_q\mathcal M:\ v \text{ future causal and } \|v\|_h=1\}.
\]
This set is compact and $dt_q(v)>0$ on $K_q$, hence $m(q):=\min_{v\in K_q}dt_q(v)>0$.
By compactness of $\overline W$, $m_0:=\min_{q\in\overline W}m(q)>0$. Thus for every future causal
$v\in T_q\mathcal M$ with $q\in\overline W$,
\begin{equation}\label{eq:dt-lower-bound}
	dt_q(v)\ge m_0\|v\|_h.
\end{equation}

Now let $r\in J^+(p)$ and choose a future-directed causal curve $\gamma$ from $p$ to $r$.
If $\gamma$ exits $W$, let $q$ be its first exit point. Parameterize $\gamma$ by $h$-arc length on the segment
inside $\overline W$. Then \eqref{eq:dt-lower-bound} gives
\[
t(q)-t(p)=\int dt(\dot\gamma)\ge \int m_0\|\dot\gamma\|_h = m_0\, L_h(\gamma|_{[p,q]})\ge m_0\, d.
\]
Hence if $t(r)-t(p)<m_0 d$, $\gamma$ cannot exit $W$, so $r\in W\subset U$.
Take $\varepsilon:=m_0 d$.
\end{proof}

\subsection{Null generators of \texorpdfstring{$\partial I^{+}(a)$}{boundary of I+(a)} in dimension 2}

We record a standard causality fact in the special case of futures of points.

\begin{lem}\label{lem:null-generator}
Let $a\in\mathcal M$ and let $q\in\partial I^+(a)\setminus\{a\}$. Then there exists a future-directed
null segment $\gamma$ from $a$ to $q$ such that $\gamma\subset\partial I^+(a)$.
\end{lem}

\begin{proof}[Proof sketch]
Choose a sequence $q_n\in I^+(a)$ converging to $q$, and for each $n$ choose a future-directed timelike curve
from $a$ to $q_n$. By the limit curve theorem in globally hyperbolic spacetimes, a subsequence converges to a
future-directed causal curve from $a$ to $q$, hence $q\in J^+(a)$.
By global hyperbolicity there exists a length-maximizing causal curve $\gamma$ from $a$ to $q$, and any such
maximizer is a causal geodesic (see \cite[Theorem~3.18]{BeemEhrlich}). Since $q\notin I^+(a)$, $\gamma$ cannot be timelike,
hence $\gamma$ is null. If an interior point of $\gamma$ lies in $I^+(a)$, then by the push-up property we obtain
$q\in I^+(a)$, a contradiction. Therefore $\gamma\subset \partial I^+(a)$.
\end{proof}

In dimension $2$ the future null cone at a point has exactly two rays, hence $\partial I^+(a)$ has exactly two
future-directed null generators issuing from $a$.

\begin{lem}\label{lem:two-generators}
For each $a\in\mathcal M$ there exist exactly two inextendible future-directed null geodesics
$\gamma_a^{-},\gamma_a^{+}\subset \partial I^+(a)$ starting at $a$ such that
\[
\partial I^+(a)\setminus\{a\} = \gamma_a^{-}\ \dot\cup\ \gamma_a^{+}.
\]
Moreover, $\gamma_a^{-}\cap\gamma_a^{+}=\{a\}$.
\end{lem}

\begin{proof}
In a convex normal neighborhood of $a$, $\partial I^+(a)$ consists of exactly two null rays emanating from $a$.
Extend those rays as maximal (inextendible) null geodesics $\gamma_a^{-},\gamma_a^{+}$.

Let $q\in\partial I^+(a)\setminus\{a\}$. By Lemma~\ref{lem:null-generator}, there is a null geodesic segment from $a$
to $q$ contained in $\partial I^+(a)$. Its initial tangent at $a$ must be one of the two future null directions, so by
uniqueness of geodesics it is contained in $\gamma_a^{-}$ or $\gamma_a^{+}$. Hence
$\partial I^+(a)\setminus\{a\}\subset \gamma_a^{-}\cup\gamma_a^{+}$.

We next show $\gamma_a^{-}\cap\gamma_a^{+}=\{a\}$. Suppose not, and let $q\neq a$ be an
intersection point of minimal parameter along $\gamma_a^{+}$. Before $q$ neither generator
coincides with the other, and null geodesics in $1+1$ dimensions have no conjugate points;
hence neither segment $\gamma_a^{\pm}|_{[a,q]}$ has a cut point before $q$, so both lie in
$\partial I^+(a)$, and so does their common endpoint $q$ by closedness of $\partial I^+(a)$.
The two segments are distinct null geodesics from $a$ to $q$, so the bigon they bound contains
a timelike curve from $a$ to $q$ (push-up inside the diamond between the two null curves),
giving $q\in I^+(a)$ and contradicting $q\in\partial I^+(a)$. Hence $\gamma_a^{-}\cap\gamma_a^{+}=\{a\}$.

By the standard null cut theorem, if one of the future null generators were to leave
$\partial I^+(a)$, then before or at the first exit there would be either a null
conjugate point or a second null geodesic from $a$ to the same point. In $1+1$
dimensions there are no transverse Jacobi fields along a null geodesic, hence no
null conjugate points; the second alternative has just been ruled out by the null-bigon
argument. Therefore each of $\gamma_a^{-},\gamma_a^{+}$ remains on $\partial I^+(a)$ for all
future parameter values. Combined with
$\partial I^+(a)\setminus\{a\}\subset\gamma_a^{-}\cup\gamma_a^{+}$, this gives
$\partial I^+(a)\setminus\{a\}=\gamma_a^{-}\ \dot\cup\ \gamma_a^{+}$.
\end{proof}

\subsection{Future intervals on slices and continuity of endpoints}

For $a\in\mathcal M$ and $s>t(a)$, define the two boundary points on $\Sigma_s$ by
\[
L_a(s):=\gamma_a^{-}\cap \Sigma_s,\qquad R_a(s):=\gamma_a^{+}\cap \Sigma_s,
\]
which are well-defined by Lemma~\ref{lem:meets-slice-once}. Define the endpoint functions
\[
\ell_a(s):=x(L_a(s)),\qquad r_a(s):=x(R_a(s)).
\]

\begin{lem}\label{lem:slice-interval}
For every $a\in\mathcal M$ and every $s>t(a)$,
\[
I^+(a)\cap \Sigma_s \;=\; \{p\in\Sigma_s:\ \ell_a(s)<x(p)<r_a(s)\}.
\]
In particular, $I^+(a)\cap \Sigma_s$ is an open interval in $\Sigma_s\simeq\mathbb R$ with endpoints $L_a(s),R_a(s)$.
\end{lem}

\begin{proof}
The set $I^+(a)\cap\Sigma_s$ is open in $\Sigma_s$ and nonempty (follow a future timelike curve from $a$ until it hits
$\Sigma_s$). The two points $L_a(s),R_a(s)$ lie in $\partial I^+(a)\cap\Sigma_s$, hence in the boundary of
$I^+(a)\cap\Sigma_s$.

We claim that $\partial I^+(a)\cap\Sigma_s=\{L_a(s),R_a(s)\}$. Indeed, any $q\in\partial I^+(a)\cap\Sigma_s$ lies on
$\partial I^+(a)$ and therefore belongs to $\gamma_a^{-}$ or $\gamma_a^{+}$ by Lemma~\ref{lem:two-generators}; since each
$\gamma_a^{\pm}$ meets $\Sigma_s$ exactly once, $q$ must be $L_a(s)$ or $R_a(s)$.

Thus $I^+(a)\cap\Sigma_s\subset\Sigma_s\simeq\mathbb R$ is a nonempty open subset whose boundary consists of exactly two
points. In a $1$-manifold, this forces $I^+(a)\cap\Sigma_s$ to be an open interval with those boundary points as endpoints,
i.e. the set of $p\in\Sigma_s$ with $\ell_a(s)<x(p)<r_a(s)$.
\end{proof}

\begin{lem}\label{lem:endpoints-continuous}
For each $a\in\mathcal M$, the functions $s\mapsto \ell_a(s)$ and $s\mapsto r_a(s)$ are continuous on $(t(a),\infty)$.
\end{lem}

\begin{proof}
We prove continuity for $\ell_a$; the proof for $r_a$ is identical.
Consider the generator $\gamma_a^{-}$. By Lemma~\ref{lem:temporal-monotone}, $t\circ\gamma_a^{-}$ is strictly increasing;
by Lemma~\ref{lem:meets-slice-once}, its image is $(t(a),\infty)$. Hence $t\circ\gamma_a^{-}$ is a homeomorphism onto
$(t(a),\infty)$, with continuous inverse.
Define $\widetilde\gamma(s):=\gamma_a^{-}\bigl((t\circ\gamma_a^{-})^{-1}(s)\bigr)\in\Sigma_s$. Then $\widetilde\gamma$
is continuous and $\ell_a(s)=x(\widetilde\gamma(s))$ is a composition of continuous maps.
\end{proof}

\bibliographystyle{amsplain}

\end{document}